\documentclass{article}
\usepackage{graphicx} 

\usepackage[utf8]{inputenc}
\usepackage[T1]{fontenc}
\usepackage{geometry}
\usepackage{mathpazo}
\usepackage{setspace}
\usepackage{amsmath, amsthm, mathtools, amssymb}
\usepackage{algorithm}
\usepackage{algpseudocode}
\usepackage{mathrsfs}
\usepackage{makecell}
\usepackage{shuffle}
\usepackage{bbm}
\usepackage{subcaption}
\usepackage{hyperref}
\usepackage{xurl}
\usepackage{placeins}
\usepackage{csquotes}
\usepackage{nomencl}
\usepackage{enumitem}
\usepackage{booktabs,tabularx,makecell,array}
\usepackage{xcolor}

\newcommand{\MC}[1]{\textsc{MC}~#1}  
\newcommand{\Q}[1]{\textsc{Q}~#1}    

\numberwithin{equation}{section}

\usepackage{authblk}

\usepackage{amsmath}
\usepackage{mathtools}
\usepackage{amsfonts}
\usepackage{amsthm}

\newtheorem{theorem}{Theorem}[section]
\newtheorem{remark}[theorem]{Remark}
\newtheorem{proposition}[theorem]{Proposition}

\newtheorem{corollary}[theorem]{Corollary}

\newtheorem{definition}[theorem]{Definition}

\newtheorem{example}[theorem]{Example}
\usepackage{graphicx}

\newcommand{\norm}[1]{\lVert#1\rVert}
\newcommand{\abs}[1]{\lvert#1\rvert}

\newcommand{\scalar}[1]{\langle #1 \rangle}

\DeclarePairedDelimiter\floor{\lfloor}{\rfloor}
\usepackage[
backend=biber,
style=apa,
sorting=nyt,
natbib=true,
doi=false
]{biblatex}
\title{A Wiener Chaos Approach to Martingale Modelling and Implied Volatility Calibration}

\author[1]{Pere Diaz-Lozano\thanks{\texttt{peredl@math.uio.no}}}
\author[2,3]{Thomas K. Kloster\thanks{\texttt{tkk@econ.au.dk}}}

\affil[1]{Department of Mathematics, University of Oslo}
\affil[2]{Department of Economics and Business Economics, Aarhus University}
\affil[3]{CoRE, Center for Research in Energy: Economics and Markets}

\hypersetup{
  pdftitle={A Wiener Chaos Approach to Martingale Modelling and Implied Volatility Calibration},
  pdfauthor={Pere Diaz-Lozano and Thomas K. Kloster}
}

\begin{document}

\maketitle
\vspace{-0.1cm}
\begin{abstract}
Calibration to a surface of option prices requires specifying a suitably flexible martingale model for the discounted asset price under a risk-neutral measure. Assuming Brownian noise, we construct an over-parameterized model for integrable nonnegative terminal prices using
the normalized square of a finite Wiener chaos expansion. We work with the expansion in its Hermite polynomial formulation, obtaining explicit expressions for the conditional expectations that determine the price process, enabling fast
calibration to a target implied-volatility surface. We illustrate the flexibility and expressive power of the resulting model through numerical experiments on both simulated and real market data.
\end{abstract}



\section{Introduction}
Option pricing models are often assessed by their ability to reproduce observed market implied volatilities of liquidly traded options. Traditionally, models are built bottom-up by specifying risk-neutral dynamics for the underlying price process, aiming for dynamics that are both reasonable and interpretable. Typically, these dynamics are chosen within a parametric family with a small number of economically and statistically meaningful parameters. If the model calibrates well to market prices, it can then be used to compute hedge ratios and to price illiquid or exotic contracts in an arbitrage-free and consistent way. There is a large literature on parametric option pricing models and their relative performance; see, e.g., \textcite{Bakshi1997,Eraker2004,AndersenFusariTodorov2015,Romer2022,math11194201} for extensive studies and comparisons.

A more recent approach that has gained considerable popularity is to model the underlying price using highly over-parameterized architectures, where their degrees of freedom are generic and not directly interpretable. Examples include signature-based expansions \citep[see, e.g.,][]{Signatures1,Signatures2,Signatures3} and neural SDE models \citep[see, e.g.,][]{NeuralSDE1,NeuralSDE2}. These model classes can calibrate well to observed market data across different market regimes, but compared to the parametric approach, it is less clear how their implied dynamics behave and how well they price exotic derivatives.

This work introduces a new over-parameterized martingale model for the discounted asset price process $(S_t)_{t\in[0,\mathcal{T}]}$. Let $(\Omega,\mathcal{F},\mathbb{Q},\mathbb{F})$ be a filtered probability space and fix a time horizon $\mathcal{T}>0$. We work under a risk-neutral measure $\mathbb{Q}$, so that the discounted price process is a $(\mathbb Q, \mathbb{F})$-martingale. That is, for any $t\le \mathcal{T}$,
\begin{gather}\label{eq:ce}
S_t = \mathbb{E}\!\left[S_\mathcal{T} \mid \mathcal{F}_t\right].
\end{gather}
{In particular, once $S_\mathcal{T}$ is specified, the process on
$[0,\mathcal{T}]$ is determined (up to modification) by \eqref{eq:ce}.
Throughout, we assume that $\mathbb{F}$ is the natural filtration of a
$d$-dimensional $\mathbb{Q}$-Brownian motion. If $S_\mathcal{T}>0$
almost surely, $S_0>0$, and $S_\mathcal{T}\in L^1(\mathcal{F}_\mathcal{T})$,
then $X\coloneqq\sqrt{S_\mathcal{T}/S_0}$ belongs to $L^2(\mathcal{F}_\mathcal{T})$. Since $S$ is a martingale,
$\mathbb{E}[X^2]=1$, and $X$ admits the Wiener chaos expansion
\begin{equation}\label{eq:chaos_expansion_intro}
X=\sum_{n\geq 0}\sum_{\lvert a\rvert=n} d_a\Phi_a,
\qquad
d_a\coloneqq a!\langle\Phi_a,X\rangle_{L^2(\mathcal{F}_\mathcal{T})},
\end{equation}
where the sum converges in $L^2(\mathcal{F}_\mathcal{T})$ and
$(\sqrt{a!}\,\Phi_a)_a$ is an orthonormal basis of
$L^2(\mathcal{F}_\mathcal{T})$. We approximate this expansion by a finite
chaos truncation $X_\theta$, with coefficient vector $\theta=(d_a)_a$,
subject to the constraint $\mathbb{E}[X_\theta^2]=1$. We define the model terminal price and price process by
\[
S_\mathcal{T}^\theta\coloneqq S_0X_\theta^2,
\qquad
S_t^\theta
\coloneqq
\mathbb{E}[S_\mathcal{T}^\theta\mid\mathcal{F}_t]
=
S_0\,\mathbb{E}[X_\theta^2\mid\mathcal{F}_t],
\quad t\in[0,\mathcal{T}].
\]
Treating the coefficients $(d_a)_a$ as trainable weights yields a flexible
model class which is universal, in the terminal $L^1$ norm, among
nonnegative integrable $(\mathbb{Q},\mathbb{F})$-martingales. We call it the
\emph{Wiener chaos martingale model}.

Consequently, the proposed class can approximate price processes generated by Brownian-driven models (e.g., Heston or SABR) in terminal $L^1$, provided their discounted prices are integrable martingales. The same applies to Volterra-type rough-volatility models with a weakly singular kernel in the volatility when the martingale condition holds.}

In practice, calibration is performed on market call option prices observed across strikes and maturities. Under the risk-neutral measure $\mathbb{Q}$, assuming zero interest rates, the discounted model price of a European call with maturity $T\le \mathcal{T}$ and strike $K$ is
\begin{gather}\label{eq:call-option-price}
    C^{\mathrm{model}}(T, K) = \mathbb{E}\!\left[(S_T-K)_+\right].
\end{gather}

{Given quoted prices
$\{C^{\mathrm{mkt}}(T,K)\,:\,(T,K)\in\mathbf{T}\times\mathbf{K}\}$,
calibration amounts to fitting the coefficients by minimizing a loss that
measures the discrepancy with the corresponding model prices. While we evaluate \eqref{eq:call-option-price} via Monte Carlo simulation,
calibration can be carried out efficiently in two steps by first pre-computing and storing offline the
quantities needed to evaluate
$\mathbb{E}[X_\theta^2\mid\mathcal{F}_t]$ on the Monte Carlo samples. The coefficients $(d_a)_a$ are then fitted using a gradient-based optimization method, subject to $\mathbb{E}[X_\theta^2]=1$.}

We conduct an extensive study of the proposed model by calibrating it to option price surfaces generated by established parametric models that can be cast in our framework, namely the Heston and rough Heston models. We find that the model fits the resulting implied-volatility surfaces accurately when the driving noise is taken to be a two-dimensional Brownian motion. To assess whether the over-parameterized model merely overfits, we evaluate out-of-sample quantities, including call prices at maturities not used in calibration and prices of several path-dependent options. In all cases, the Wiener chaos martingale model reproduces these prices with reasonable accuracy, suggesting that the learned dynamics remain economically plausible and useful for pricing beyond the European option surface. As a final test, we demonstrate that the model also calibrates efficiently and accurately to real SPX option data.

\subsection{Related Literature}
Before proceeding, we briefly comment on previous works that employ chaos expansions for option pricing and hedging.

In \textcite{Chaos_Funahashi}, the authors develop a closed-form approximation for European option prices in Markovian diffusion models of the form $dS_t = S_t\,\sigma(S_t,t)\,dW_t$, where $\sigma$ is a deterministic local-volatility function. Their method combines a Wiener chaos expansion with a successive substitution scheme, yielding explicit low-order formulas for the characteristic function of the normalized return, which is then integrated to obtain call prices. They illustrate the approach on parametric local-volatility specifications such as CEV. In \textcite{Chaos_BasketOption}, this methodology is extended to multi-asset settings for basket-type payoffs. Our perspective differs in that we do not start from a parametric diffusion and derive an expansion from it; instead, we treat the chaos expansion as the model itself and calibrate the expansion coefficients directly.

In \textcite{neufeld2025chaotichedgingiteratedintegrals}, the market is taken as given and the goal is to hedge contingent claims, rather than to fit a model to observed option data. Starting from a general exponentially integrable semimartingale $S$, the authors show that any contingent claim that is an $L^p$-functional of $S$ can be approximated arbitrarily well by a finite number of terms in the iterated-integral chaos expansion of $S$. They then approximate the associated kernels with neural networks and use this representation to solve the $L^p$ hedging problem.

A paper closer in spirit to our approach is \textcite{Chaos_RiskNeutralDistribution}. There, the marginal distribution of $S_t$ is modeled via a truncated Wiener chaos expansion in a Hermite basis, and the coefficients are calibrated to option data separately at each maturity. Because calibration is performed independently across maturities, the resulting collection of risk-neutral marginals need not be consistent with a single arbitrage-free price process. In contrast, our framework starts from a single chaos-driven model and derives all maturities from it in a time-consistent and arbitrage-free manner.

\subsection{Structure of the paper}

{The paper is organized as follows. Section~\ref{sec:preliminaries} introduces the notation and recalls the main ingredients of Wiener chaos expansions used throughout. Section~\ref{sec:model} presents the Wiener chaos martingale model and establishes its universality over integrable nonnegative Brownian martingales. In Section~\ref{sec:conditional_expectation}, we derive explicit expressions for the conditional square of the chaos factor, which are then used in the calibration procedure of Section~\ref{sec:calibration}. Section~\ref{sec:numerical-experiments} reports numerical experiments, including calibration to synthetic option surfaces generated by standard parametric models and a range of robustness checks. Finally, we fit the model to SPX option data.}

\section{Preliminaries}\label{sec:preliminaries}

\subsection{Definitions and notation}

Let $\mathcal{T} > 0$ be a fixed time horizon, and let $B = (B_{t})_{t \in [0,\mathcal{T}]}$ be a $d$-dimensional Brownian motion defined on a complete probability space $(\Omega, \mathcal{F}, \mathbb{Q})$. We use the following spaces and notational conventions:
\begin{itemize}
    \item $\mathbb{F} = (\mathcal{F}_{t})_{ t \in [0,\mathcal{T}]}$ is the filtration generated by the Brownian motion $B$ and augmented with the $\mathbb{Q}$-null sets. We assume that $\mathcal{F} = \mathcal{F}_{\mathcal{T}}$.

    \item $\mathbb{E}[\ \cdot \mid \mathcal{F}_t]$ denotes conditional expectation under $\mathbb{Q}$, and $\mathbb{E}[\cdot]$ the corresponding expectation.

    \item $L^{b}(\mathcal{F}_t)$, $t \in [0, \mathcal{T}]$ and $b \geq 1$, is the space of all $\mathcal{F}_{t}$-measurable random variables $X\colon \Omega\to \mathbb{R}$ satisfying $\norm{X}_{b} \coloneqq \mathbb{E}\! [ \abs{X}^{b} ]^{1/b}  < \infty$.

{    \item $\mathcal{M}_{\mathcal{T}}^{1}$ denotes the space of all continuous,
real-valued $(\mathbb{Q},\mathbb{F})$-martingales $M$ on $[0,\mathcal{T}]$
such that $\| M \|_{\mathcal M_\mathcal{T}^1} \coloneqq \mathbb{E}\!\left[|M_\mathcal{T}|\right]<\infty$.}

    \item $C_{p}^{\infty}(\mathbb{R}^{n \times d})$ is the space of smooth functions $\phi: \mathbb{R}^{n \times d} \to \mathbb{R}$ such that $\phi$ and all its partial derivatives have polynomial growth.
\end{itemize}

We also recall some basic definitions related to Malliavin calculus.

\begin{itemize}

    \item $\mathbb{H}_{\mathcal{T}} = L^{2}([0,\mathcal{T}]; \mathbb{R}^{d})$ is the space of functions $h\colon[0,\mathcal{T}] \to \mathbb{R}^{d}$ such that 
    \begin{gather*}
        \norm{h}_{\mathcal{T}} \coloneqq \Big(\int_{0}^{\mathcal{T}} |h(s)|^{2} ds \Big)^{1/2} < \infty.
    \end{gather*}
    
    \item For $h = (h^{1}, \dots, h^{d}) \in \mathbb{H}_{\mathcal{T}}$, we define the random vector 
    \begin{gather*}
        \mathbf{B}(h) \coloneqq \Big(\int_{0}^{\mathcal{T}} h^{1}(s) dB_{s}^{1}, \dots, \int_{0}^{\mathcal{T}} h^{d}(s) dB_{s}^{d} \Big).
    \end{gather*}

   \item $\mathcal{S}_{\mathcal{T}}$ denotes the class of smooth random variables that are $\mathcal{F}_{\mathcal{T}}$-measurable, which have the form 
   \begin{gather*}
       F = \phi(\mathbf{B}(h_{1}), \dots, \mathbf{B}(h_{n})),
   \end{gather*}
   where the function $\phi((x_{1}^{1}, \dots, x_{1}^{d}), \dots, (x_{n}^{1}, \dots, x_{n}^{d}))$ belongs to $C_{p}^{\infty}(\mathbb{R}^{n \times d})$, and for all $i \leq n$ we have $h_{i} = (h_{i}^{1}, \dots, h_{i}^{d}) \in \mathbb{H}_{\mathcal{T}}$.

   \item The Malliavin derivative of a random variable $F \in \mathcal{S}_{\mathcal{T}}$ is defined as the $d$-dimensional stochastic process
   \begin{gather*}
       D_{s}^{j} F = \sum_{i=1}^{n} \frac{\partial \phi}{\partial x_{i}^{j}} \big( \mathbf{B} (h_{1}), \dots, \mathbf{B} (h_{n}) \big) h_{i}^{j}(s), \quad s \in [0,\mathcal{T}], \quad j=1,\dots,d.
   \end{gather*}

   \item More generally, for $k \geq 2$, we define the $k$-th Malliavin derivative of $F \in \mathcal{S}_{\mathcal{T}}$ along the direction $\alpha = (\alpha_{1}, \dots, \alpha_{k}) \in \{1, \dots, d\}^{k}$ evaluated at $\mathbf{s} = (s_{1}, \dots, s_{k}) \in [0,\mathcal{T}]^{k}$ as the real-valued random variable
   \begin{gather*}
        D_{\mathbf{s}}^{\alpha} F \coloneqq \sum_{i_{1}, \dots, i_{k} = 1}^{n} \frac{\partial^{k} \phi}{\partial x_{i_{1}}^{\alpha_{1}} \cdots \partial x_{i_{k}}^{\alpha_{k}}} \big( \mathbf{B}(h_{1}), \dots, \mathbf{B}(h_{n}) \big) h_{i_{1}}^{\alpha_{1}}(s_{1}) \cdots h_{i_{k}}^{\alpha_{k}}(s_{k}).
    \end{gather*}
\end{itemize}

\subsection{Wiener chaos expansion and Hermite polynomials}
We now present an elementary introduction to the Wiener chaos expansion of $L^2(\mathcal{F}_{\mathcal{T}})$ using Hermite polynomials. We refer to \textcite[Section 1.1]{NualartDavidTMCa} for more details on this topic.

Let $(H_n)_{n\geq0}$ denote the Hermite polynomials defined by
\begin{gather}\label{eq:hermite-exponential}
\sum_{n=0}^{\infty}H_n(x)\eta^n
=
\exp\!\left(x\eta-\frac{\eta^2}{2}\right).
\end{gather}
Differentiating with respect to $x$ and comparing coefficients gives
$H_n'(x)=H_{n-1}(x)$ for $n\geq1$, with $H_0(x)=1$.
We fix an orthonormal basis $(h_i)_{i\geq 1}$ of
$L^{2}([0,\mathcal{T}];\mathbb{R})$ and denote by
$(e_j)_{1\leq j\leq d}$ the canonical basis of $\mathbb{R}^d$.
Then the family
\begin{gather*}
    (e_jh_i)_{i\geq 1,\,1\leq j\leq d}
\end{gather*}
is an orthonormal basis of $L^{2}([0,\mathcal{T}];\mathbb{R}^d)$.
Consider the set of multi-indices $a = (a^1, \dots, a^d)$, where
$a^j = (a_{1}^j, a_{2}^j, \dots)$ is a sequence of nonnegative integers
such that all the terms, except a finite number of them, vanish.
Set $a! = \prod_{j=1}^{d} \prod_{i \geq 1} a_i^j!$ and
$|a| = \sum_{j=1}^{d} \sum_{i \geq 1} a_i^j$, and define
\begin{gather}\label{eq:phi_a-def}
    \Phi_{a} = \prod_{j=1}^d \prod_{i \geq 1} H_{a_i^j}
    \Big(\int_{0}^{\mathcal{T}} h_i(s)\,dB_s^j\Big).
\end{gather}
Notice that the infinite product in \eqref{eq:phi_a-def} is well defined
because $H_0\equiv 1$ and $a_i^j\neq 0$ only for a finite number of indices.

For each $n \geq 0$, let $\mathcal{H}_{\mathcal{T}}^{n}$ denote the closed linear subspace of $L^{2}(\mathcal{F}_{\mathcal{T}})$ generated by the random variables $\{ \Phi_{a}, |a| = n\}$. This space is called the Wiener chaos of order $n$. We then have the following orthogonal expansion, see Theorem 1.1.1 and Proposition 1.1.1 in \textcite{NualartDavidTMCa}.

\begin{theorem}
   \begin{enumerate}
       \item The space $L^{2}(\mathcal{F}_{\mathcal{T}})$ can be decomposed into the infinite orthogonal sum of the subspaces $\mathcal{H}_{\mathcal{T}}^{n}$.
    \item The collection of random variables $\{ \sqrt{a!}  \Phi_{a}, |a| = n \}$ forms a complete orthonormal system in $\mathcal{H}_{\mathcal{T}}^{n}$. 
   \end{enumerate} 
\end{theorem}

We therefore deduce the following result, called the Wiener chaos expansion of $L^{2}(\mathcal{F}_{\mathcal{T}})$. 
\begin{theorem}
    For any $F \in L^{2}(\mathcal{F}_{\mathcal{T}})$, we have the expansion 
    \begin{gather}
        F  = \sum_{k \geq 0} \sum_{|a|=k} a!  \scalar{\Phi_a, F}_{L^{2}(\mathcal{F}_\mathcal{T})} \Phi_{a},\label{chaos expansion hermite}
    \end{gather}
    where the sum converges in $L^{2}(\mathcal{F}_{\mathcal{T}})$.
\end{theorem}

We conclude this section by introducing the finite-dimensional subspace of $L^{2}(\mathcal{F}_\mathcal{T})$ that will serve as the foundation for our model.

\begin{definition}
Let $P, M \in \mathbb{N}$, and let us define the index set given by 
\begin{gather*}
{    \mathcal{A}_{P, M, d} \coloneqq \big\{ a = (a^1, \dots, a^d) \in \mathbb N_0^{M\times d} \mid 0 \leq |a| \leq P \big\}.}
\end{gather*} 
We then define the finite-dimensional space
\begin{gather*}
{   L^2(\mathcal{F}_\mathcal{T}) \supset \mathcal{H}_\mathcal{T}^{(P, M)} \coloneqq \text{span} \big\{ \Phi_a \mid a \in \mathcal{A}_{P, M, d}  \big\}.}
\end{gather*}
The dimension of this space can be obtained by using Stars and Bars and Pascal's identity \citep[see, e.g.,][]{FellerWilliam1968Aitp}
\begin{gather*}
    \sum_{k=0}^{P} \frac{(Md+k-1)!}{(Md -1)!k!} = \frac{(M d+ P)!}{(Md)! P!}.
\end{gather*}
\end{definition}

\section{The Wiener Chaos Martingale Model}\label{sec:model}

With the Wiener chaos expansion in place, we can now define a martingale model for the (discounted) underlying asset price process. 

Similarly to many other works, we choose to model the price process directly under a risk-neutral measure $\mathbb{Q}$. Recall that $\mathbb{Q}$ is called risk-neutral if it is equivalent to the physical measure $\mathbb{P}$ and the discounted asset price process is a local $(\mathbb{Q},\mathbb{F})$--martingale. Under the standard no-arbitrage notion of no free lunch with vanishing risk (NFLVR), the existence of a risk-neutral probability measure implies the absence of arbitrage; see, e.g., \textcite{DelbaenSchachermayer1994,GuasoniRasonyiScachermayer2008}. For ease of notation, we assume zero interest rates and dividend yields in the theoretical development, so the asset price coincides with its discounted value.

We therefore propose the following martingale model for the price process, which we call the \emph{Wiener chaos martingale model}.

\begin{definition}[Wiener chaos martingale model]\label{def:model}
Fix an initial price $S_0>0$, $P,M,d\in\mathbb N$, and an orthonormal
family $(h_i)_{1\leq i\leq M}$ in $L^2([0,\mathcal T])$. For coefficients
$\theta=(d_a)_{a\in\mathcal A_{P,M,d}}$ satisfying
$\sum_{a\in\mathcal A_{P,M,d}}d_a^2/a!=1$, set
\begin{equation}\label{eq:X_theta_and_terminal_price}
X_\theta\coloneqq\sum_{a\in\mathcal A_{P,M,d}}d_a\Phi_a,
\qquad
S_\mathcal T^\theta\coloneqq S_0X_\theta^2.
\end{equation}
The model price process is the square-integrable
$(\mathbb Q,\mathbb F)$-martingale
\begin{equation}\label{eq:market-model}
S_t^\theta
\coloneqq \mathbb E[S_\mathcal T^\theta\mid\mathcal F_t]
= S_0\,\mathbb E[X_\theta^2\mid\mathcal F_t],
\qquad t\in[0,\mathcal T].
\end{equation}
\end{definition}

\begin{remark}[Normalization of the coefficients]
For any nonzero coefficient vector $\theta$, let
$\bar\theta=\theta/\sqrt{\mathbb E[X_\theta^2]}$. Then $\bar\theta$
satisfies the constraint in Definition~\ref{def:model}, and
\[
S_t^{\bar\theta}
=
S_0\,\frac{\mathbb E[X_\theta^2\mid\mathcal F_t]}
              {\mathbb E[X_\theta^2]},
\qquad t\in[0,\mathcal T].
\]
The expression on the right is unchanged when all coefficients are
multiplied by the same nonzero constant. We can therefore impose the
constraint during calibration by rescaling the coefficients after each
optimization step.
\end{remark}

{%
The model parameters are the coefficient vector
$\theta=(d_a)_{a\in\mathcal A_{P,M,d}}\in\Theta_{P, M, d}$, where
\[
\Theta_{P, M, d}\coloneqq
\left\{
(d_a)_{a\in\mathcal A_{P,M,d}}
\in\mathbb R^{|\mathcal A_{P,M,d}|}
:
\sum_{a\in\mathcal A_{P,M,d}}\frac{d_a^2}{a!}=1
\right\}.
\]
This constraint ensures $S_0^\theta=S_0$. Since $|\mathcal A_{P,M,d}|=\binom{Md+P}{P}$, the model has $\binom{Md+P}{P}$ coefficients subject to one constraint, giving $\binom{Md+P}{P}-1$ continuous degrees of freedom. The coefficient
vectors $\theta$ and $-\theta$ yield the same price process. The truncation levels $P$, $M$ and the basis functions $(h_i)_{1\le i\le M}$ are fixed in advance and treated as hyperparameters. We stress that a pricing model based on \eqref{eq:market-model} is arbitrage-free, since $S_t^\theta$ is a $(\mathbb{Q},\mathbb{F})$-martingale by construction.}

{
\begin{remark}
The model defined in \eqref{eq:market-model} is strictly positive almost surely. Indeed, it is a nonzero polynomial of Gaussian random variables and therefore
\[
\mathbb{Q}\big(X_\theta=0\big)=0.
\]
Hence, $S_\mathcal{T}^\theta>0$ almost surely by construction, and consequently $S_t^\theta>0$ almost surely for every $t\in[0,\mathcal{T}]$.
\end{remark}
}

\subsection{Universality property}

{
The following proposition gives terminal $L^1$ approximation for nonnegative martingales in $\mathcal{M}_\mathcal{T}^1$.
}
{
\begin{proposition}\label{prop:universality}
Assume $S \in \mathcal{M}_\mathcal{T}^1$, with $S_t\ge0$ almost surely for every $t\in[0,\mathcal{T}]$ and $S_0>0$, and let $(h_i)_{i\ge1}$ be a complete orthonormal system in $L^2([0,\mathcal{T}])$. Then, for any $\epsilon>0$, there exist $P,M\in\mathbb N$ and $\theta\in\Theta_{P, M, d}$ such that
\begin{gather}\label{eq:universality}
\|S-S^\theta\|_{\mathcal{M}_\mathcal{T}^1}
=\mathbb E\!\left[|S_\mathcal{T}-S_\mathcal{T}^\theta|\right]
\le\epsilon.
\end{gather}
\end{proposition}
}
{
\begin{proof}
Set $X=\sqrt{S_\mathcal{T}/S_0}$. Since $S_\mathcal{T}\in L^1$, we have $X\in L^2$ and $\mathbb E[X^2]=1$. Finite Wiener chaos polynomials are dense in $L^2$, so there exist finite chaos polynomials $\widetilde X_n$ with truncation levels $P_n,M_n$ such that $\widetilde X_n\to X$ in $L^2$. Since $\|\widetilde X_n\|_{L^2}\to1$, for all sufficiently large $n$ we can set
\[
X_{\theta_n}=\frac{\widetilde X_n}{\|\widetilde X_n\|_{L^2}}.
\]
Then $\theta_n\in\Theta_{P_n, M_n, d}$ and $X_{\theta_n}\to X$ in $L^2$. By Cauchy--Schwarz,
\[
\big\|X_{\theta_n}^2-X^2\big\|_{L^1}
\leq \|X_{\theta_n}-X\|_{L^2}
\bigl(\|X_{\theta_n}\|_{L^2}+\|X\|_{L^2}\bigr)\longrightarrow0.
\]
Consequently, $S_0X_{\theta_n}^2\to S_\mathcal{T}$ in $L^1$. The result follows by choosing $n$ sufficiently large.
\end{proof}
}

{The terminal $L^1$ estimate directly controls European call prices. It also controls path-dependent payoffs with a sublinear H\"older modulus.

\begin{proposition}
Assume that \eqref{eq:universality} holds. Let $T \le \mathcal{T}$ and $F : C([0,T]) \to \mathbb{R}$ satisfy, for some $\alpha \in (0,1)$ and $L>0$,
\begin{gather}\label{assmpt:holder-functional}
    |F(x) - F(y)| \leq L \|x-y\|_\infty^{\alpha}
\qquad \text{for all } x,y \in C([0,T]).
\end{gather}
Then
\begin{gather*}
    \Big|
\mathbb{E} \big[F((S_t)_{t\leq T})\big]
-
\mathbb{E}\big[F((S_t^\theta)_{t\leq T})\big]
\Big|
\leq \frac{L}{1-\alpha}\,\epsilon^\alpha.
\end{gather*}
\end{proposition}

\begin{proof}
Set $D_t=S_t-S_t^\theta$. By conditional Jensen and \eqref{eq:universality}, $\mathbb E[|D_T|]\leq\mathbb E[|D_\mathcal T|]\leq\epsilon$. Doob's weak $L^1$ maximal inequality gives, for every $\lambda>0$,
\[
\mathbb Q\!\left(\sup_{t\leq T}|D_t|>\lambda\right)
\leq \min\!\left(1,\frac{\epsilon}{\lambda}\right).
\]
Integrating this tail bound yields
\[
\mathbb E\!\left[\sup_{t\leq T}|D_t|^\alpha\right]
\leq \alpha\int_0^\epsilon\lambda^{\alpha-1}\,d\lambda
+\alpha\epsilon\int_\epsilon^\infty\lambda^{\alpha-2}\,d\lambda
=\frac{\epsilon^\alpha}{1-\alpha}.
\]
The claim follows from \eqref{assmpt:holder-functional}.
\end{proof}

For a European call, the Lipschitz payoff gives the direct bound
\[
\left|\mathbb E[(S_T-K)_+]-\mathbb E[(S_T^\theta-K)_+]\right|
\leq\mathbb E[|S_T-S_T^\theta|]\leq\epsilon.
\]}

\subsection{Dynamic properties of the model}\label{sec:dynamic-properties}

For fixed $P,M,d$, the terminal price $S_\mathcal{T}^\theta$ is a polynomial in Gaussian variables and hence has finite moments of all orders. By conditional Jensen's inequality, the same holds for $S_t^\theta$ at each $t$. To describe the dynamics, we split the Brownian integrals into their observed and future parts:
\[
Y_t^{(i,j)}\coloneqq\int_0^t h_i(s)\,dB_s^j,
\qquad
Z_t^{(i,j)}\coloneqq\int_t^\mathcal{T} h_i(s)\,dB_s^j,
\qquad 1\leq i\leq M,\quad 1\leq j\leq d.
\]
Let $Y_t$ and $Z_t$ be the corresponding vectors, with the coordinates for each $i$ grouped together. We write $X_\theta=\psi_\theta(Y_\mathcal{T})=\psi_\theta(Y_t+Z_t)$, where $\psi_\theta$ is the polynomial of degree at most $P$ given by \eqref{eq:X_theta_and_terminal_price}. The vector $Z_t$ is independent of $\mathcal F_t$ and centered Gaussian, with covariance
\[
\Sigma_t\coloneqq G(t)\otimes I_d,
\qquad
G_{i\ell}(t)\coloneqq\int_t^\mathcal{T}h_i(s)h_\ell(s)\,ds.
\]

\begin{proposition}\label{prop:S_dynamics}
Let $\sigma_Y(t)=h(t)\otimes I_d\in\mathbb R^{Md\times d}$, where $h(t)=(h_1(t),\ldots,h_M(t))^\top$. Then
\begin{equation}\label{eq:S_markovian}
S_t^\theta=u_\theta(t,Y_t),
\qquad
dY_t=\sigma_Y(t)\,dB_t,
\qquad Y_0=0,
\end{equation}
where $u_\theta(t,y)\coloneqq S_0\mathbb E[\psi_\theta(y+Z_t)^2]$. For each $t$, $u_\theta(t,\cdot)$ is a polynomial of degree at most $2P$, and its coefficients are absolutely continuous in $t$. It satisfies the backward equation
\begin{equation}\label{eq:u_backward_repr}
\begin{aligned}
\partial_t u_\theta(t,y)
+\frac{1}{2}\operatorname{tr}\!\left(
\sigma_Y(t)\sigma_Y(t)^\top \nabla_y^2u_\theta(t,y)
\right)&=0
\qquad\text{for a.e. }t\in[0,\mathcal{T}],\\
u_\theta(\mathcal{T},y)&=S_0\psi_\theta(y)^2.
\end{aligned}
\end{equation}
Consequently, the price process satisfies
\begin{equation}\label{eq:S_SDE_repr2}
dS_t^\theta=\nabla_y u_\theta(t,Y_t)^\top\sigma_Y(t)\,dB_t.
\end{equation}
\end{proposition}

\begin{proof}
The representation $S_t^\theta=u_\theta(t,Y_t)$ follows from \eqref{eq:market-model} and the independence of $Z_t$ and $\mathcal F_t$, and the terminal condition follows from $Z_\mathcal{T}=0$. The coefficients of $u_\theta(t,\cdot)$ are polynomials in the entries of $\Sigma_t$, since the Gaussian moments are polynomials in those entries. They are therefore absolutely continuous in $t$, as
\[
\Sigma_t=\int_t^\mathcal{T}\sigma_Y(s)\sigma_Y(s)^\top\,ds.
\]
To see the backward equation, consider the Gaussian generating function
\[
\mathbb E\!\left[e^{\eta^\top(y+Z_t)}\right]
=\exp\!\left(\eta^\top y+\frac{1}{2}\eta^\top\Sigma_t\eta\right),
\qquad \eta\in\mathbb R^{Md}.
\]
Differentiating in $t$ gives a factor $-\tfrac12\eta^\top\sigma_Y(t)\sigma_Y(t)^\top\eta$ for almost every $t$, while the Hessian in $y$ gives a factor $\eta\eta^\top$. The generating function therefore satisfies the differential equation in \eqref{eq:u_backward_repr}, and comparing coefficients in $\eta$ yields the same equation for $u_\theta$. Applying It\^o's formula term by term to $u_\theta(t,Y_t)$ gives \eqref{eq:S_SDE_repr2}. Since $Y$ and the polynomial coefficients are continuous, this representation gives a continuous version of $S^\theta$.
\end{proof}

The Clark--Ocone formula gives the same diffusion coefficient. Indeed, by using that $D_t(X_\theta^2)=2X_\theta D_tX_\theta$, we obtain
\begin{equation}\label{eq:S_SDE_repr1}
\sigma_Y(t)^\top\nabla_y u_\theta(t,Y_t)
=2S_0\mathbb E\!\left[X_\theta D_tX_\theta\mid\mathcal F_t\right]
\end{equation}
for almost every $t$, almost surely.

\begin{corollary}
The price process admits the multiplicative representation
\[
dS_t^\theta=S_t^\theta\,\nu_\theta(t,Y_t)^\top\,dB_t,
\qquad
\nu_\theta(t,y)=
\frac{\sigma_Y(t)^\top\nabla_y u_\theta(t,y)}{u_\theta(t,y)}
\quad\text{when }u_\theta(t,y)>0.
\]
\end{corollary}

\begin{proof}
Since $S^\theta$ is a nonnegative continuous martingale with $S_\mathcal{T}^\theta>0$ almost surely, optional sampling shows that it cannot hit zero. The result then follows from \eqref{eq:S_SDE_repr2}.
\end{proof}

\begin{remark}
Since $Y$ has independent increments, \eqref{eq:S_markovian}
provides a finite-dimensional Markov state for the model.
The diffusion coefficient generally varies with time through
$\sigma_Y(t)$ and the remaining covariance $G(t)$.
Such time dependence also appears in local-stochastic volatility models,
which augment traditional stochastic volatility models with a time-dependent
local volatility correction, and in the extended short-rate models of
\cite{HullWhite19990}, which add a time-dependent drift to calibrate
exactly to observed bond prices.
\end{remark}

\section{Computation of conditional expectations}\label{sec:conditional_expectation}

{
We now turn to the issue of computing the conditional expectations appearing in \eqref{eq:market-model}, namely
\begin{gather}\label{eq:squared-cond-exp}
    \mathbb E\!\big[|X_\theta|^2\mid\mathcal F_t\big].
\end{gather}
We begin with a general representation formula for the conditional second moment of a finite Wiener chaos expansion.
}

{
\begin{proposition}\label{prop:conditional-square}
Let $X\in \bigoplus_{n=0}^{P} \mathcal H_\mathcal{T}^{n}$. Then, for every $t\in[0,\mathcal T]$,
\[
\mathbb E[X^2\mid\mathcal F_t]
=
\left|\mathbb E[X\mid\mathcal F_t]\right|^2
+
\sum_{k=1}^{P}\frac{1}{k!}
\sum_{\alpha\in\{1,\ldots,d\}^k}
\int_{[t,\mathcal T]^k}
\left|
\mathbb E\!\left[D^\alpha_s X\mid\mathcal F_t\right]
\right|^2 ds.
\]
\end{proposition}
}

{
\begin{proof}
By the Clark--Ocone formula,
\[
X=\mathbb E[X\mid\mathcal F_t]
+\sum_{j=1}^d\int_t^\mathcal T
\mathbb E[D_s^jX\mid\mathcal F_s]\,dB_s^j.
\]
Hence, by the conditional It\^o isometry,
\[
\mathbb E[X^2\mid\mathcal F_t]
=
\left|\mathbb E[X\mid\mathcal F_t]\right|^2
+
\sum_{j=1}^d\int_t^\mathcal T
\mathbb E\!\left[
\left|\mathbb E[D_s^jX\mid\mathcal F_s]\right|^2
\Bigm|\mathcal F_t
\right]ds.
\]
Applying the same argument recursively to the Malliavin derivatives yields
\[
\mathbb E[X^2\mid\mathcal F_t]
=
\left|\mathbb E[X\mid\mathcal F_t]\right|^2
+
\sum_{k=1}^{P}
\sum_{\alpha\in\{1,\ldots,d\}^k}
\int_{t<s_k<\cdots<s_1<\mathcal T}
\left|
\mathbb E[D^\alpha_sX\mid\mathcal F_t]
\right|^2ds.
\]
Since Malliavin derivatives commute, permuting the pairs $(s_i,\alpha_i)$ leaves the integrand unchanged. After summing over $\alpha$, the integrand is therefore symmetric in $s_1,\ldots,s_k$, so its integral over the simplex is $1/k!$ times its integral over $[t,\mathcal T]^k$. The iteration stops at $k=P$ since $X$ has chaos order at most $P$.
\end{proof}
}

{
We now apply Proposition~\ref{prop:conditional-square} to $X_\theta$ in \eqref{eq:X_theta_and_terminal_price}.
With $Y_t$, $Z_t$, and $\psi_\theta$ as in Section~\ref{sec:dynamic-properties}, define
\[
\mu_t(y)\coloneqq
\mathbb E\!\left[\psi_\theta(y+Z_t)\right].
\]
Since $Y_t$ is $\mathcal F_t$-measurable and $Z_t$ is independent of $\mathcal F_t$, the freezing lemma yields $\mu_t(Y_t)=\mathbb E[X_\theta\mid\mathcal F_t]$.

Here $\partial_i^j$ denotes differentiation with respect to the coordinate $y_i^j$. For $\alpha=(\alpha_1,\ldots,\alpha_k)\in\{1,\ldots,d\}^k$ and $s=(s_1,\ldots,s_k)$, the Malliavin chain rule, the freezing lemma, and differentiation under the Gaussian expectation give
\[
\begin{aligned}
\mathbb E\!\left[D_s^\alpha X_\theta\mid\mathcal F_t\right]
&=
\sum_{i_1,\ldots,i_k=1}^M
\mathbb E\!\left[
\partial_{i_1}^{\alpha_1}\cdots\partial_{i_k}^{\alpha_k}
\psi_\theta(Y_t+Z_t)\mid\mathcal F_t
\right]
\prod_{r=1}^k h_{i_r}(s_r)\\
&=
\sum_{i_1,\ldots,i_k=1}^M
\partial_{i_1}^{\alpha_1}\cdots
\partial_{i_k}^{\alpha_k}\mu_t(Y_t)
\prod_{r=1}^k h_{i_r}(s_r).
\end{aligned}
\]
}

{
Substituting into Proposition~\ref{prop:conditional-square} and integrating the basis functions gives
\begin{flalign}\notag
\mathbb E\!\left[|X_\theta|^2\mid\mathcal F_t\right] \\
&=
\sum_{k=0}^P\frac{1}{k!}
\sum_{\alpha\in\{1,\ldots,d\}^k}
\sum_{\substack{i_1,\ldots,i_k=1\\
                  \ell_1,\ldots,\ell_k=1}}^M
\left(\prod_{r=1}^k G_{i_r\ell_r}(t)\right)
\partial_{i_1}^{\alpha_1}\cdots\partial_{i_k}^{\alpha_k}\mu_t(Y_t)\,
\partial_{\ell_1}^{\alpha_1}\cdots\partial_{\ell_k}^{\alpha_k}\mu_t(Y_t).
\label{eq:formula-square-x}
\end{flalign}

For $\beta\in\mathcal A_{P,M,d}$, write
$\partial^\beta=\prod_{j=1}^d\prod_{i=1}^M(\partial_i^j)^{\beta_i^j}$.
For each choice of indices in \eqref{eq:formula-square-x}, set
\[
\beta_i^j
=\#\{r:i_r=i,\ \alpha_r=j\},
\qquad
\gamma_i^j
=\#\{r:\ell_r=i,\ \alpha_r=j\}.
\]
Then $|\beta|=|\gamma|=k$, and the two derivative factors are
$\partial^\beta\mu_t(Y_t)$ and $\partial^\gamma\mu_t(Y_t)$, respectively.
Since $\psi_\theta$ is a polynomial, differentiation under the expectation
and $H_n'=H_{n-1}$ give, for every $\beta\in\mathcal A_{P,M,d}$,
\begin{gather}\label{eq:partial_beta_m_t}
\begin{aligned}
\partial^\beta\mu_t(y)
&=\mathbb E\!\left[\partial^\beta\psi_\theta(y+Z_t)\right]=
\sum_{\substack{a\in{\mathcal A_{P,M,d}}\\a\geq\beta}}
d_a\,\mathbb E\!\left[
\prod_{j=1}^d\prod_{i=1}^M
H_{a_i^j-\beta_i^j}\bigl(y_i^j+Z_t^{(i,j)}\bigr)
\right].
\end{aligned}
\end{gather}
where $a\geq\beta$ is understood coordinatewise. Each expectation is a
finite Gaussian moment calculation determined by $G(t)$.
In both cases below, we evaluate these derivatives and substitute them
into \eqref{eq:formula-square-x} to obtain $S_t^\theta$ through
\eqref{eq:market-model}.
}

\subsection{The case of piecewise-constant functions}\label{sec:piecewises_constant}

Let $0=s_0<s_1<\cdots<s_M=\mathcal T$ be a partition of $[0,\mathcal T]$, set $\delta_i=s_i-s_{i-1}$, and choose
\[
h_i(s)=\frac{\mathbf 1_{(s_{i-1},s_i]}(s)}{\sqrt{\delta_i}},
\qquad i=1,\ldots,M.
\]
These functions have disjoint supports and form an orthonormal family, which can be completed, for example, by Haar functions within each subinterval. This choice has been used in \textcite{BriandPhilippe2014SOBB,lozano2025eulerschemebsdeswiener}. Fix $t\in(s_{u-1},s_u]$. Then, the future Gram matrix is diagonal:
\[
G_{i\ell}(t)=\delta_{i\ell}g_i(t),
\qquad
g_i(t)=
\begin{cases}
0, & i<u,\\
1- \frac{t-s_{u-1}}{\delta_u}, & i=u,\\
1, & i>u.
\end{cases}
\]
In particular, the variance of $Y_t^{(i,j)}$ is $v_i(t)\coloneqq 1-g_i(t)$ for any $1 \leq j \leq d$.

For $v\geq0$, let $H_n(y;v)$ denote the variance-$v$ Hermite polynomial, defined by
\begin{gather}\label{eq:hermite-v-exponential}
\sum_{n=0}^{\infty}H_n(y;v)\eta^n
=
\exp\!\left(y\eta-\frac{v\eta^2}{2}\right).
\end{gather}
Thus $H_n(y;1)=H_n(y)$ and
$\partial_y H_n(y;v)=H_{n-1}(y;v)$ for $n \geq 1$. We will use the Gaussian identity
\begin{gather}\label{eq:gaussian-hermite-identity}
\mathbb E[H_n(y+Z;1)]=H_n(y;v),
\qquad Z\sim N(0,1-v),\quad v\in[0,1].
\end{gather}
To see this, take the expectation of the generating function and use
the moment-generating function of $Z$, given by $\mathbb E[e^{\eta Z}]=\exp((1-v)\eta^2/2)$, to see that
\[
\begin{aligned}
\mathbb E\!\left[
\exp\!\left((y+Z)\eta-\frac{\eta^2}{2}\right)
\right]
&=
\exp\!\left(y\eta-\frac{v\eta^2}{2}\right).
\end{aligned}
\]
By using \eqref{eq:hermite-exponential}, the coefficient of $\eta^n$ on the left is $\mathbb E[H_n(y+Z;1)]$, while by using \eqref{eq:hermite-v-exponential} we obtain that the one on the right is $H_n(y;v)$.

Set $g(t)^\beta\coloneqq\prod_{j=1}^d\prod_{i=1}^M
g_i(t)^{\beta_i^j}$, with $0^0=1$.

\begin{corollary}\label{crl:conditional-expectation-model}
For $t\in(s_{u-1},s_u]$, the squared Wiener chaos model satisfies
\begin{gather}\label{eq:conditional-expectation-model}
S_t^\theta
=
S_0\sum_{\beta\in\mathcal A_{P,M,d}}
\frac{g(t)^\beta}{\beta!}
\left|\partial^\beta\mu_t(Y_t)\right|^2,
\end{gather}
where, for $\beta\in\mathcal A_{P,M,d}$,
\[
\partial^\beta\mu_t(y)
=
\sum_{\substack{a\in{\mathcal A_{P,M,d}}\\
                a\geq\beta}}
d_a\prod_{j=1}^d\prod_{i=1}^M
H_{a_i^j-\beta_i^j}\bigl(y_i^j;v_i(t)\bigr).
\]
\end{corollary}

\begin{proof}
Since $G_{i\ell}(t)=\delta_{i\ell}g_i(t)$, only terms with
$\ell_r=i_r$ for every $r$ survive in \eqref{eq:formula-square-x}.
The two derivative multi-indices therefore coincide, and the product
of Gram matrix entries is $g(t)^\beta$. For $|\beta|=k$, there are
$k!/\beta!$ ordered lists of pairs $(i_r,\alpha_r)$ with these
multiplicities. Combining this count with the factor $1/k!$ gives
the weight $g(t)^\beta/\beta!$, including $\beta=0$ when $k=0$.
Equation~\eqref{eq:market-model} then yields
\eqref{eq:conditional-expectation-model}.

Finally, independence of the future coordinates
$(Z_t^{(i,j)})_{i,j}$ allows each expectation in
\eqref{eq:partial_beta_m_t} to factorize over $(i,j)$.
Applying the Gaussian identity
\eqref{eq:gaussian-hermite-identity} to each factor gives
the stated expression for $\partial^\beta\mu_t$.
\end{proof}

\begin{remark}\label{rmk:conditional-expectation-model}
\begin{enumerate}
\item The weight $g(t)^\beta$ automatically restricts the sum to
multi-indices satisfying $\beta_i^j=0$ whenever $g_i(t)=0$, that is, on every interval completed by time $t$. For future intervals ($i>u$), we have $Y_t^{(i,j)}=v_i(t)=0$, so a term indexed by $a$ in $\partial^\beta\mu_t(Y_t)$ survives only when $\beta_i^j=a_i^j$. Thus only the current interval can split a Hermite degree between its observed and future parts.

\item The value of $S_t^\theta$ depends on at most $ud$ independent
Gaussian coordinates: the completed Brownian increments and the
observed part of the current increment. This makes low-dimensional
Gaussian quadrature possible at short maturities.

\item The expression in \eqref{eq:conditional-expectation-model}
is a polynomial of degree at most $2P$ in those observed coordinates.
Its moments can therefore be obtained from finite Gaussian
moment calculations.
\end{enumerate}
\end{remark}

\subsection{The general case}
For a general orthonormal family, $G(t)$ need not be diagonal,
so the expectations in \eqref{eq:partial_beta_m_t} need not
factorize. Define the differential operator
\[
\Delta_{G(t)}
\coloneqq
\sum_{j=1}^d\sum_{i,\ell=1}^M
G_{i\ell}(t)\,\partial_i^j\partial_\ell^j.
\]
Applying the Dyson formula of \textcite[Eq.~(2.20)]{Jin03072016}
to $\partial^\beta\psi_\theta(y+Z_t)$ gives, for
$\beta\in\mathcal A_{P,M,d}$,
\begin{equation}\label{eq:partial_beta_general}
\partial^\beta\mu_t(y)
=
\sum_{n=0}^{\floor{(P-|\beta|)/2}}
\frac{1}{2^n n!}\,
\Delta_{G(t)}^n
\bigl(\partial^\beta\psi_\theta\bigr)(y).
\end{equation}
Here $\Delta_{G(t)}^n$ denotes the $n$-fold composition of $\Delta_{G(t)}$. The sum terminates because $\partial^\beta\psi_\theta$ has degree at most $P-|\beta|$, and each application of $\Delta_{G(t)}$ involves two derivatives. Substituting these expressions into \eqref{eq:formula-square-x} and using \eqref{eq:market-model}
gives $S_t^\theta$.

For example, for $P=2$, we have
$\mu_t=\psi_\theta+\tfrac12\Delta_{G(t)}\psi_\theta$.
Let $\nabla$ and $\nabla^2$ denote the gradient and Hessian in the
coordinate order of $\Sigma_t$.
The terms $k=0,1,2$ in \eqref{eq:formula-square-x} then yield
\begin{equation*}
\begin{aligned}
S_t^\theta
&=S_0\Bigl[\mu_t(Y_t)^2
+\nabla\mu_t(Y_t)^\top\Sigma_t\nabla\mu_t(Y_t)+\frac12\operatorname{tr}\!\left(
\Sigma_t\nabla^2\mu_t(Y_t)\Sigma_t\nabla^2\mu_t(Y_t)
\right)\Bigr].
\end{aligned}
\end{equation*}

\begin{example}[Legendre basis]\label{ex:Legendre}
An example of an orthonormal basis of $L^2([0,\mathcal{T}])$, which we use in the numerical experiments, can be constructed from the Legendre polynomials. These are polynomials $\{\mathcal{L}_n\}_{n\ge 0}$ on $[-1,1]$ defined by 
\[
\mathcal{L}_n(x)
=
\frac{1}{2^n n!}\frac{d^n}{dx^n}\bigl[(x^2-1)^n\bigr],
\qquad x\in[-1,1],
\]
and they satisfy the orthogonality relation $\int_{-1}^1 \mathcal{L}_n(x)\,\mathcal{L}_m(x)\,dx =
\frac{2}{2n+1}\,\delta_{nm}$. If we define, for $i\ge 1$,
\[
h_i(s)
:=
\sqrt{\frac{2i-1}{\mathcal{T}}}\;
\mathcal{L}_{\,i-1}\!\left(\frac{2s}{\mathcal{T}}-1\right),
\qquad s\in[0,\mathcal{T}],
\]
then $\{h_i\}_{i\ge 1}$ forms an orthonormal basis of $L^2([0,\mathcal{T}])$.
\end{example}

\section{Calibration procedure}\label{sec:calibration}

In Section~\ref{sec:conditional_expectation}, we derived expressions for the conditional expectations of the chaos basis $\{\Phi_a : a \in \mathcal{A}_{P, M, d}\}$, which yield an explicit representation of the Wiener chaos martingale model via \eqref{eq:market-model}. In this section, we explain how this can be calibrated to observed market data.

A standard approach to calibration under the risk-neutral measure $\mathbb{Q}$ is to fit the model to prices of liquidly traded European plain-vanilla options, either calls or puts. We calibrate the model to European call options, whose payoff is given by $(S_T - K)_+$, where $T$ and $K$ denote the maturity and the strike, respectively. The arbitrage-free price at time $t=0$ in our model is given by
\begin{gather}\label{eq:price_call_model}
    C_\theta^{\mathrm{model}}(T, K) = \mathbb{E}\!\big[ (S_T^{\theta} - K)_+ \big].
\end{gather}

Given $N$ observations of market call prices with maturity–strike pairs $(T_i,K_i)$,
\[
    C^{\mathrm{mkt}}(T_1, K_1), \dots, C^{\mathrm{mkt}}(T_N, K_N),
\]
we follow the standard approach of \textcite{Signatures1,ChristoffersenHestonJacobs2009} and formulate calibration as the optimization problem
\begin{gather}\label{eq:opt_problem}
    \theta^{*} \in \operatorname*{argmin}_{\theta \in \Theta_{P, M, d}} \sum_{i=1}^{N} \gamma_i 
    \Big( C^{\mathrm{mkt}}(T_i, K_i) - C_\theta^{\mathrm{model}}(T_i, K_i) \Big)^2,
\end{gather}
where $\gamma_i$ denotes the Vega weight associated with the $i$-th option. These are defined by
\[
\gamma_i \coloneqq \frac{1}{\mathrm{Vega}_i^{2}},
\qquad 
\mathrm{Vega}_i \coloneqq \frac{\partial}{\partial \sigma}\,C^{\mathrm{BS}}(T_i,K_i,\sigma)\Big|_{\sigma=\sigma_i^{\mathrm{mkt}}},
\]
where $C^{\mathrm{BS}}(T,K,\sigma)$ denotes the Black--Scholes call price as a function of volatility and $\sigma_i^{\mathrm{mkt}}$ is the market implied volatility corresponding to $(T_i,K_i)$.

The optimization problem \eqref{eq:opt_problem} is typically solved using an iterative scheme, such as stochastic gradient descent (SGD). As a result, calibration requires repeated evaluations of model option prices for many values of $\theta \in \Theta_{P, M, d}$. It is therefore crucial to compute the expectations in \eqref{eq:price_call_model} efficiently. In the following, we present two approaches for doing so: The first is broadly applicable, while the second is tailored to short maturities and the choice of piecewise-constant basis functions introduced in Section~\ref{sec:piecewises_constant}.

\subsection{Monte Carlo pricing}\label{sec:MC}

The most direct way to approximate \eqref{eq:price_call_model} is by Monte Carlo simulation,
\begin{gather}\label{eq:price_call_model_MC}
    C_\theta^{\mathrm{model}}(T, K)
    = \mathbb{E} \big[ (S_T^{\theta} - K)_+ \big]
    \approx \frac{1}{n_{MC}} \sum_{j=1}^{n_{MC}} (S_T^{\theta}(\omega_j) - K)_+ .
\end{gather}
{A key observation is that the sampled observed Brownian coordinates $Y_T(\omega_j)$ and the future covariance $G(T)$ do not depend on $\theta$, so they can be computed before calibration. During optimization, the same Monte Carlo paths can be reused at each iteration. The model values $S_T^\theta(\omega_j)$ are evaluated from the conditional square in \eqref{eq:market-model}, using the coefficient vector and the normalization $\mathbb E[|X_\theta|^2]$. The conditional-square numerator is quadratic in the coefficients, while the normalized price is invariant under their common rescaling.}

\subsubsection{Control variates for variance reduction}

We use control variates to reduce the variance of the Monte Carlo
estimator; see \textcite{GlassermanPaul2004MCMi}. For fixed $(T,K)$ and
$\theta\in\Theta_{P,M,d}$, we use the centered controls
\[
X_1=S_T^\theta-S_0,
\qquad
X_2=(S_T^\theta)^2-m_2(T),
\qquad
m_2(T)=\mathbb E[(S_T^\theta)^2],
\]
giving the estimator
\begin{gather}\label{eq:MC-estimator-variance-reduction-2}
\frac{1}{n_{MC}}\sum_{j=1}^{n_{MC}}
\left[
(S_T^\theta(\omega_j)-K)_+
-\widehat\beta_1X_1(\omega_j)
-\widehat\beta_2X_2(\omega_j)
\right].
\end{gather}
The coefficients are estimated separately for each strike using regularized sample covariances; an independent fitting sample preserves unbiasedness.

For the piecewise-constant basis, expanding
$X_\theta^2=\sum_{\gamma\in\mathcal A_{2P,M,d}}q_\gamma\Phi_\gamma$
and using Hermite orthogonality gives
\begin{equation}\label{eq:second-moment-price-cv}
m_2(T)
=
S_0^2\sum_{\gamma\in\mathcal A_{2P,M,d}}
\frac{q_\gamma^2}{\gamma!}\,v(T)^\gamma,
\end{equation}
where $v(T)^\gamma=\prod_{j=1}^d\prod_{i=1}^M
v_i(T)^{\gamma_i^j}$. The coefficients $q_\gamma$ are computed using the Hermite product
formula. For general bases, $m_2(T)$ is evaluated using Gaussian moments with observed covariance $(I_M-G(T))\otimes I_d$.

\subsection{Quadrature-based pricing}

As explained in Remark~\ref{rmk:conditional-expectation-model},
when the basis functions are piecewise-constant on a time grid
\[
0=s_0<s_1<\cdots<s_M=\mathcal T,
\]
the payoff $(S_T^\theta-K)_+$ depends on $ud$ independent Gaussian
coordinates, where $T\in(s_{u-1},s_u]$.

For $z\in\mathbb R^{u\times d}$, let
$y_T(z)^{(i,j)}=\sqrt{v_i(T)}z_i^j$ for $i\leq u$ and zero
otherwise, where $v_i(T)=1-g_i(T)$. By
Corollary~\ref{crl:conditional-expectation-model}, define
\[
g_\theta(z)\coloneqq
S_0\sum_{\beta\in\mathcal A_{P,M,d}}
\frac{g(T)^\beta}{\beta!}
\left|\partial^\beta\mu_T\bigl(y_T(z)\bigr)\right|^2.
\]
We can then express
\[
\mathbb E\big[(S_T^\theta-K)_+\big]
=
\int_{\mathbb R^{u\times d}}
\bigl(g_\theta(z)-K\bigr)_+\,
f_{(u,d)}(z)\,dz,
\]
where $f_{(u,d)}$ denotes the standard Gaussian density.

When $ud$ is relatively small, we approximate this integral
using Gauss--Hermite quadrature
\citep[see, e.g.,][]{HildebrandFrancisBegnaud1956Itna}.
This approach is particularly useful at short maturities and
can also help with deep out-of-the-money options, where Monte
Carlo can struggle because exercise is rare.

\section{Numerical experiments}\label{sec:numerical-experiments}

In this section, we calibrate the Wiener chaos martingale model introduced in Section~\ref{sec:model} to option price data. We consider synthetic prices generated by the Heston and rough Heston models, as well as real SPX option data. Throughout, we compare the results obtained using the piecewise-constant basis of Section~\ref{sec:piecewises_constant} and the Legendre basis of Example~\ref{ex:Legendre}. We use the calibration procedure described in Section~\ref{sec:calibration} with the implementation details provided in Appendix~\ref{sec:calibration-details}.

The code used in the numerical experiments is available at \url{https://github.com/pere98diaz/A-Wiener-Chaos-Approach-to-Martingale-Modelling-and-Implied-Volatility-Calibration}.

\subsection{The Heston model}\label{subsec:heston_model}

The classical Heston model of \textcite{Heston1993} is a stochastic volatility model in which the log-price $X=\log(S)$ has the risk-neutral dynamics
\begin{flalign}
    dX_t &= -\tfrac{1}{2}V_t\,dt
    + \sqrt{V_t}\left(\rho\,dB_t^1+\sqrt{1-\rho^2}\,dB_t^2\right),
    \label{eq:Heston-model}\\
    dV_t &= \kappa(\overline{v}-V_t)\,dt
    + \varepsilon\sqrt{V_t}\,dB_t^1,
    \quad V_0=v_0. \notag
\end{flalign}
Here $B^1$ and $B^2$ are independent Brownian motions, so that $\rho$ is the correlation between the Brownian drivers of the log-price and variance. The characteristic function of $X_t$ is available in closed form; see, e.g., \textcite{Gatheral2006}.

\subsubsection{Fit to option implied volatilities}

We consider the Heston model \eqref{eq:Heston-model} with parameters
\begin{equation}\label{eq:heston_params}
S_0 = 100, \quad \kappa = 1.5, \quad \bar{v} = 0.04, \quad
\varepsilon = 0.5, \quad \rho = -0.7, \quad V_0 = 0.04,
\quad r=q=0.
\end{equation}

For the Wiener chaos martingale model, we set the time horizon to $\mathcal{T}=1.974$. We use $d=2$ Brownian drivers. For the piecewise-constant basis, we take $M=7$ sub-intervals on $[0,\mathcal{T}]$; for the Legendre basis, we use $M=10$ functions per driver. In both cases, we truncate the chaos expansion at order $P=2$. This yields 120 and 231 coefficients, respectively, or 119 and 230 effective continuous degrees of freedom after normalization.

We calibrate to call option prices on a grid of seven maturities, from one month to two years, and nine strikes spanning a moneyness range of $\pm 20\%$, following \textcite{Signatures1}. To assess interpolation across maturities, we also evaluate the fit at six additional maturities that are not used in calibration.

We run 100 independent calibrations and report in Table~\ref{tab:error_stats_heston} the mean absolute error (MAE) of the implied-volatility smiles and the calibration time, given as average $\pm$ standard deviation. The errors are reported in basis points.

\begin{table}[!htbp]
\centering
\begin{tabular}{|l|cc|c|}
\hline
& \multicolumn{2}{c|}{MAE (bps)}
& Calibration Time (s) \\
\hline
Basis
& Calibrated & Uncalibrated
& Average $\pm$ Std \\
\hline
piecewise-constant
& $4.08 \pm 0.64$ & $10.80 \pm 0.83$ & $17.26 \pm 2.74$ \\
Legendre
& $5.44 \pm 0.91$ & $5.71 \pm 0.93$ & $19.93 \pm 3.17$ \\
\hline
\end{tabular}
\caption{Statistics over 100 calibrations to Heston option prices. We report the mean absolute error (MAE) of the implied-volatility surfaces in basis points (average $\pm$ standard deviation) for calibrated and uncalibrated maturities. Calibration time includes Heston initialization and chaos calibration.}
\label{tab:error_stats_heston}
\end{table}

In Figures~\ref{fig:heston_fits} and \ref{fig:heston_fits_oos}, we show a representative fitted implied-volatility surface together with the corresponding surface of absolute errors. Figure~\ref{fig:heston_fits} displays the calibrated maturities, while Figure~\ref{fig:heston_fits_oos} shows the same comparison for the uncalibrated maturities.

\begin{figure}[!htbp]
\setlength{\parskip}{0pt}
\centering
	
	\begin{subfigure}[b]{0.48\textwidth}
		\includegraphics[width=\linewidth]{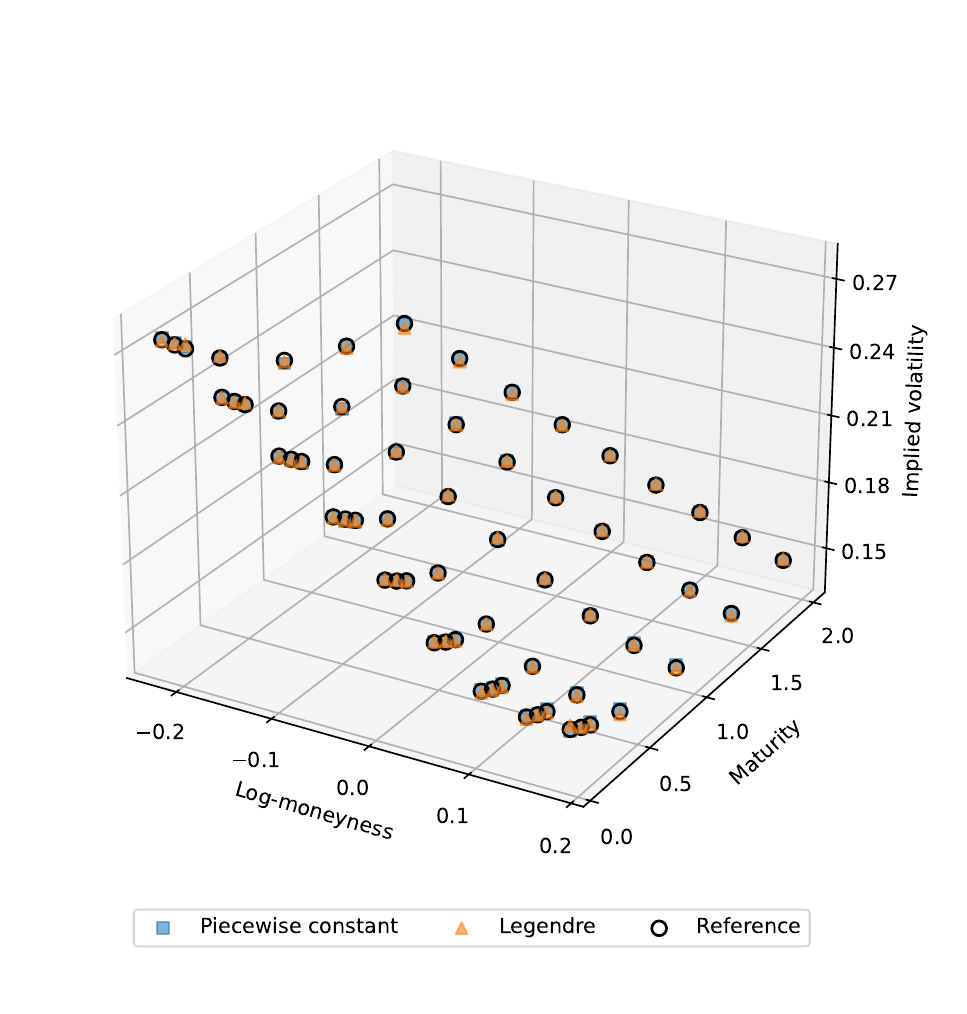}
	\end{subfigure}
	\begin{subfigure}[b]{0.48\textwidth}
		\includegraphics[width=\linewidth]{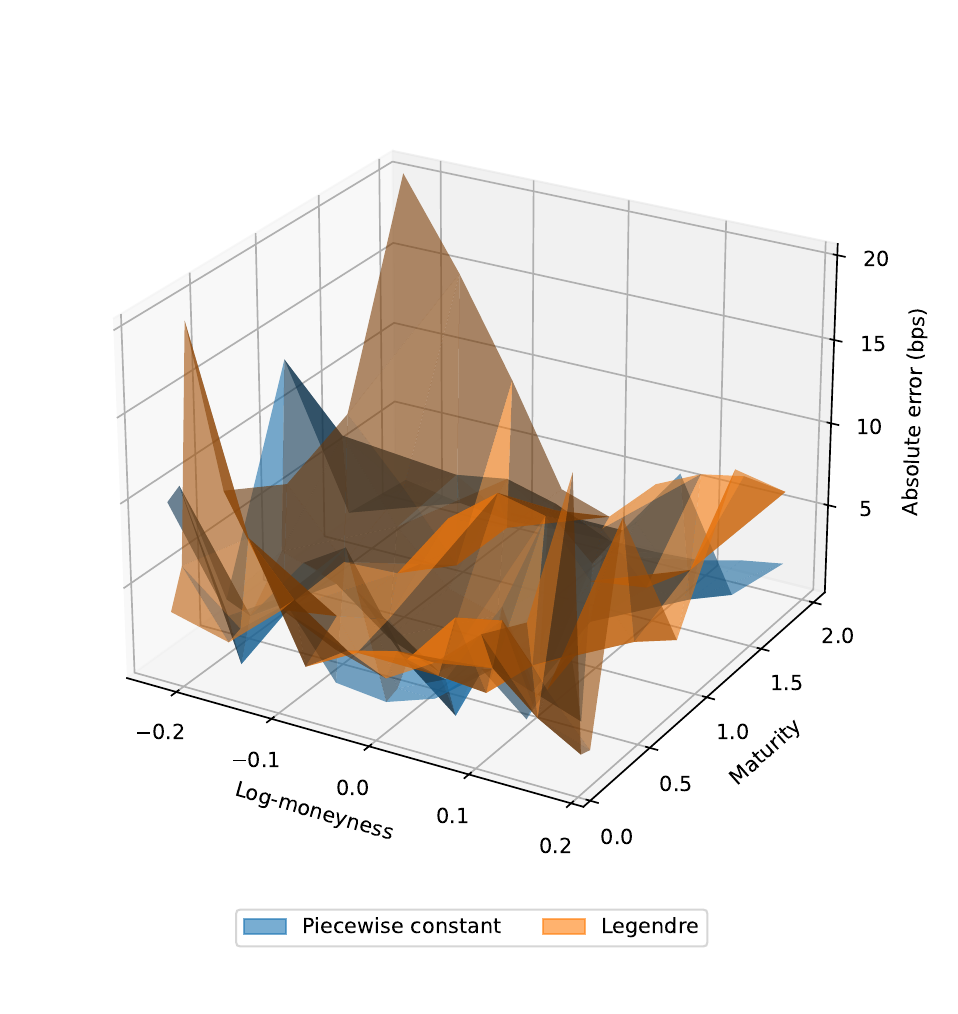}
	\end{subfigure}
	
	\vspace{0.1cm}

	\begin{subfigure}[b]{0.96\textwidth}
		\includegraphics[width=\linewidth]{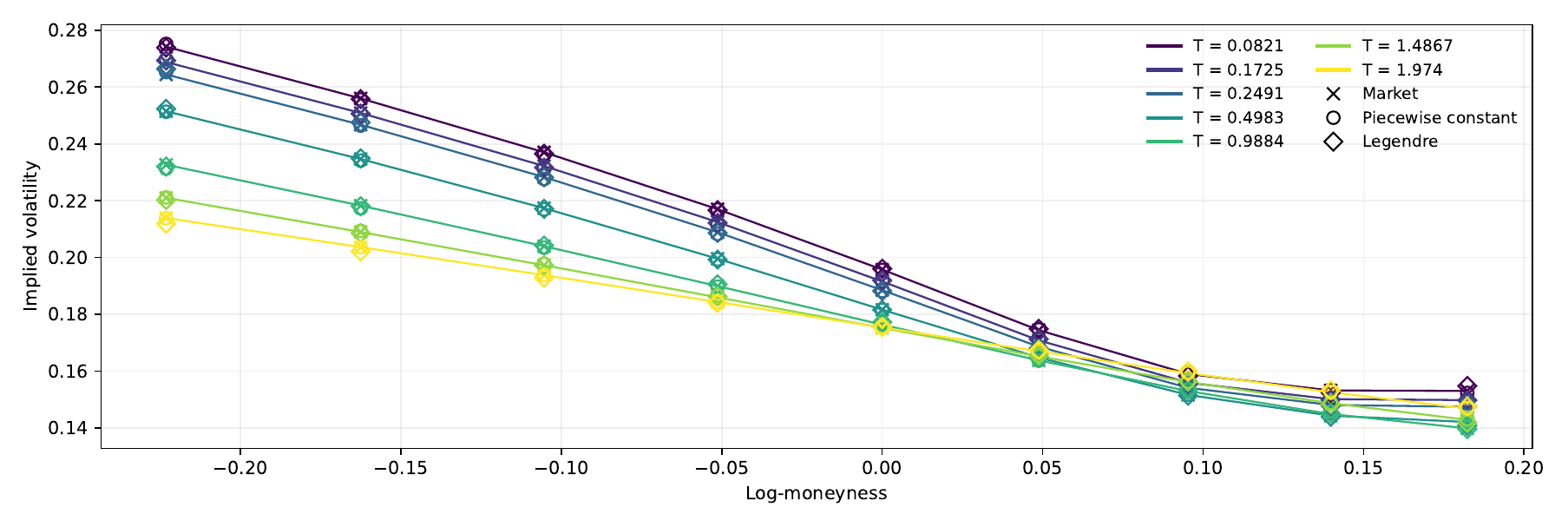}
	\end{subfigure}

	\caption{Implied-volatility surfaces at the calibrated maturities for the Heston model (top left) and the corresponding absolute-error surfaces (top right), obtained with the piecewise-constant and Legendre bases. The bottom panel shows the implied-volatility smiles. The MAE is 4.34 bps for the piecewise-constant basis and 5.18 bps for the Legendre basis.}
\label{fig:heston_fits}

\end{figure}

\begin{figure}[!htbp]
\setlength{\parskip}{0pt}
\centering
	
	\begin{subfigure}[b]{0.48\textwidth}
		\includegraphics[width=\linewidth]{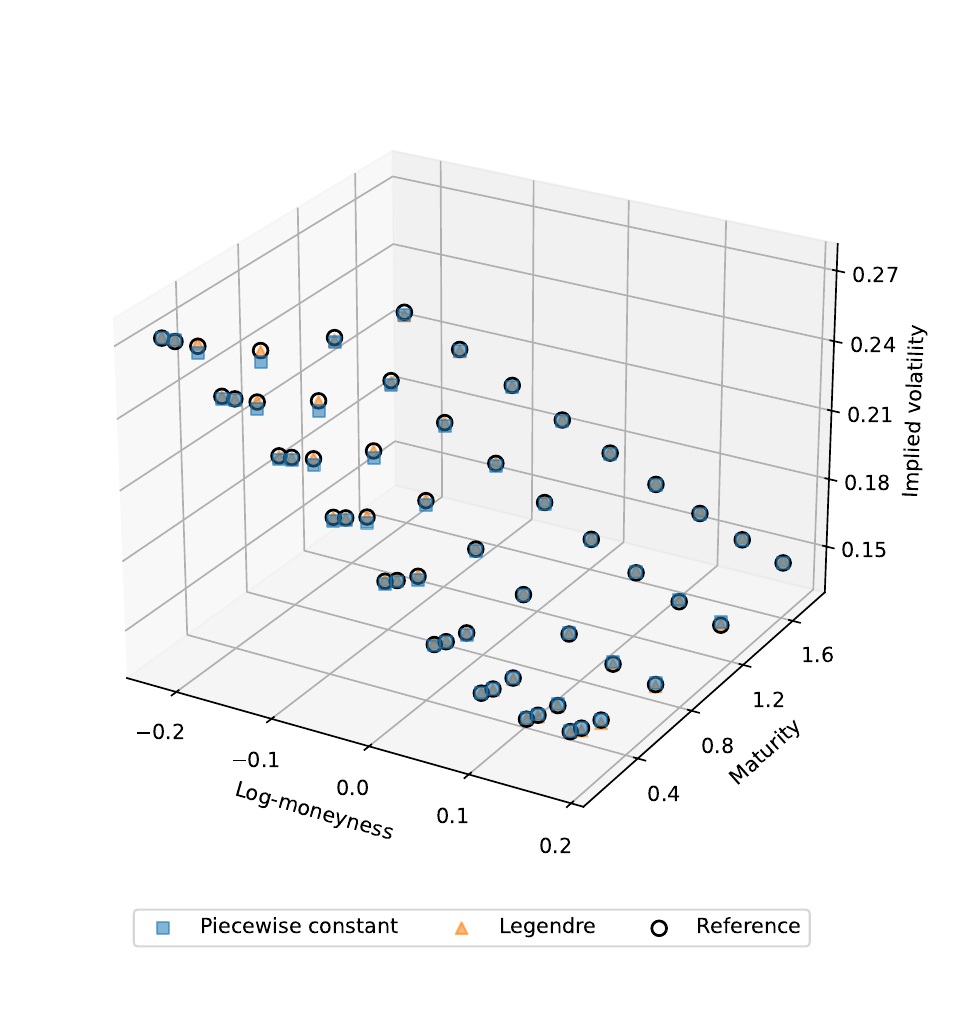}
	\end{subfigure}
	\begin{subfigure}[b]{0.48\textwidth}
		\includegraphics[width=\linewidth]{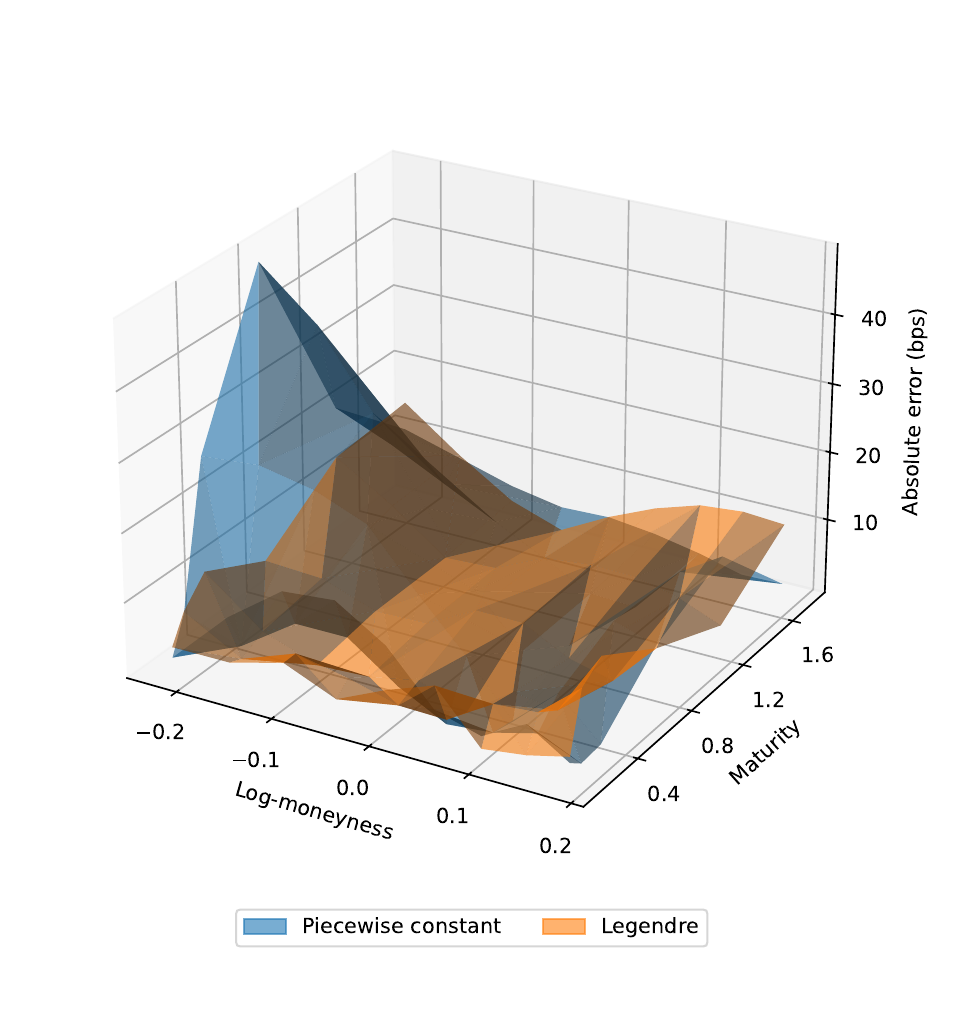}
	\end{subfigure}
	
	\vspace{0.1cm}

	\begin{subfigure}[b]{0.96\textwidth}
		\includegraphics[width=\linewidth]{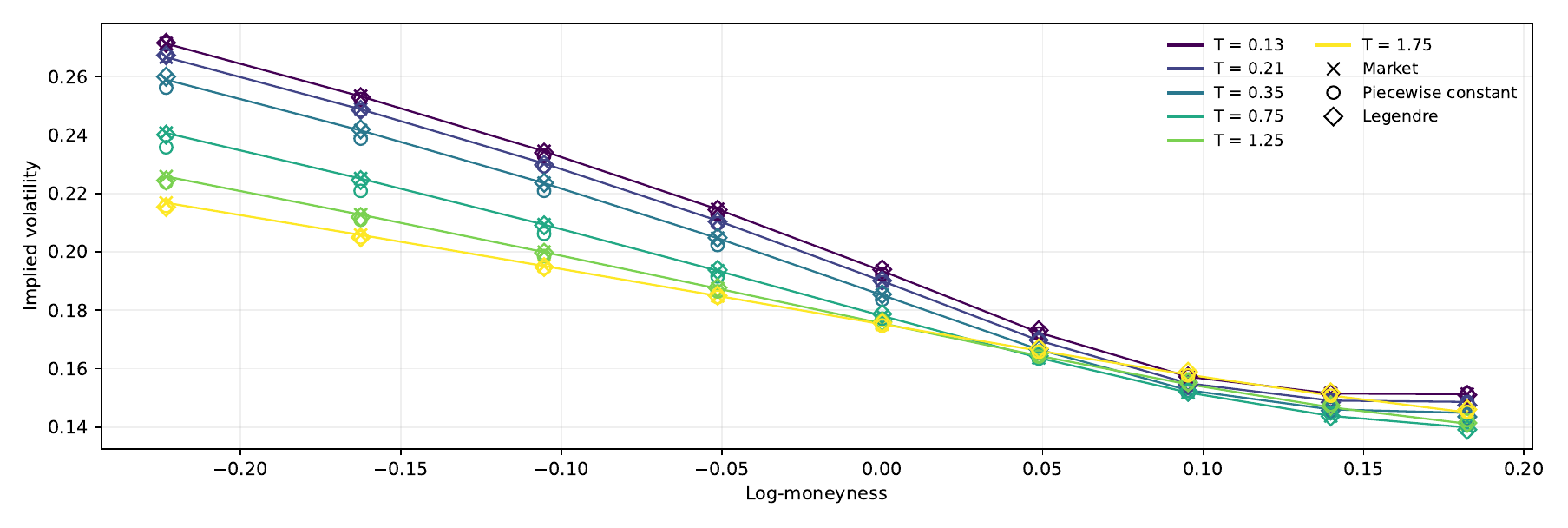}
	\end{subfigure}

	\caption{Implied-volatility surfaces at the uncalibrated maturities for the Heston model (top left) and the corresponding absolute-error surfaces (top right), obtained with the piecewise-constant and Legendre bases. The bottom panel shows the implied-volatility smiles. The MAE is 10.29 bps for the piecewise-constant basis and 5.20 bps for the Legendre basis.}\label{fig:heston_fits_oos}
\end{figure}

For the calibrated maturities, both specifications fit the Heston implied-volatility surface very well, and the two fitted surfaces are visually indistinguishable. The piecewise-constant basis achieves a lower MAE at the calibrated maturities. At the uncalibrated maturities, the Legendre basis performs better, suggesting better generalization beyond the calibration grid. Finally, the calibration times are comparable for the two bases.

\subsubsection{Out-of-sample pricing of path-dependent options}

The Wiener chaos martingale model reproduces the calibrated option prices accurately, but its many parameters raise the question of whether it overfits these prices while producing unrealistic path dynamics. Its performance at uncalibrated maturities suggests that it captures meaningful features of the marginal distributions rather than merely fitting noise. However, static implied-volatility smiles identify only the marginal distributions at fixed times and do not determine the path dynamics.

A more demanding test is therefore to compare the prices of path-dependent options under the calibrated model with those under the reference model. Such tests can be challenging even for models that fit European option prices accurately. For instance, local-volatility models can match the European option surface almost exactly yet misprice exotic derivatives because they fail to capture the relevant path dynamics. Similarly, \textcite{SchoutensPerfectCalibration} shows that models with comparably good fits to plain-vanilla options can produce widely different exotic prices.

We consider forward-starting, down-and-out, and floating-strike lookback calls, reporting implied volatilities obtained by numerically inverting the closed-form Black--Scholes prices in Appendix~\ref{sec:exotic-iv}. We restrict this analysis to the Heston model, for which reference prices are computed efficiently by Monte Carlo using the quadratic--exponential scheme of \textcite{AndersenHeston2008}. For barrier and lookback options, path extrema are monitored on a grid with spacing $0.001$ years, while the Black--Scholes formulas assume continuous monitoring. The plots use separate individual calibrations of the two bases.

\begin{figure}[!htbp]
\setlength{\parskip}{0pt}
\centering
	
	\begin{subfigure}[b]{0.90\textwidth}
		\includegraphics[width=\linewidth,trim=0 10 0 9,clip]{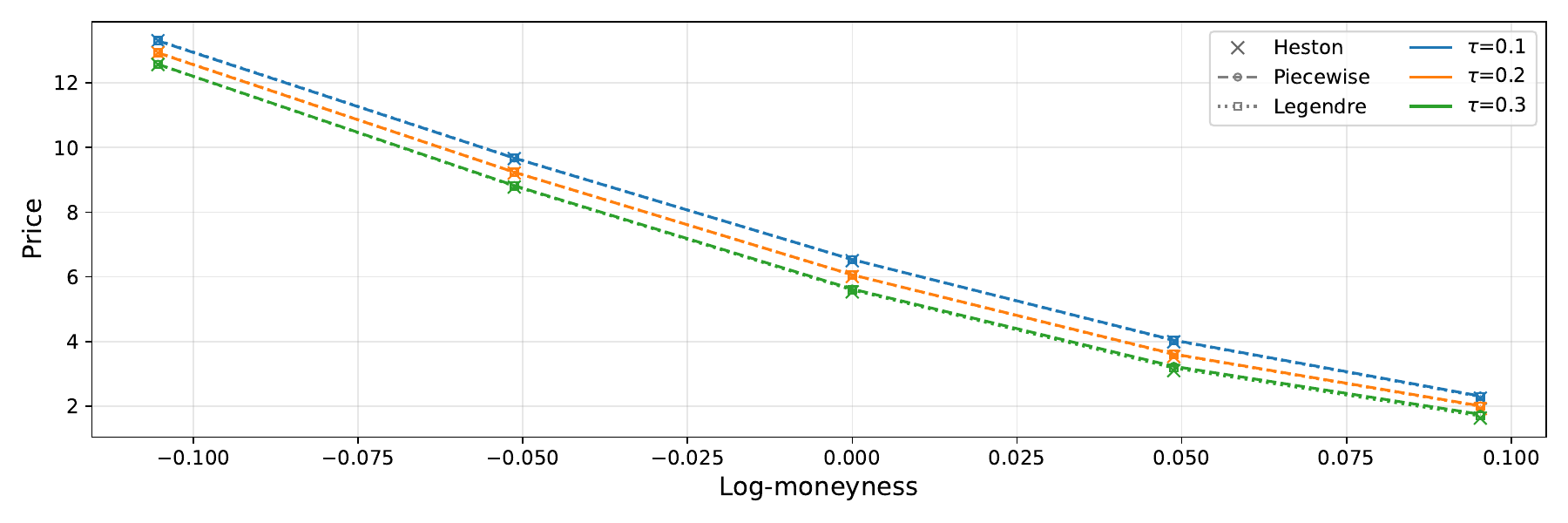}
	\end{subfigure}

    \vspace{0.1cm}
	\begin{subfigure}[b]{0.90\textwidth}
		\includegraphics[width=\linewidth,trim=0 10 0 9,clip]{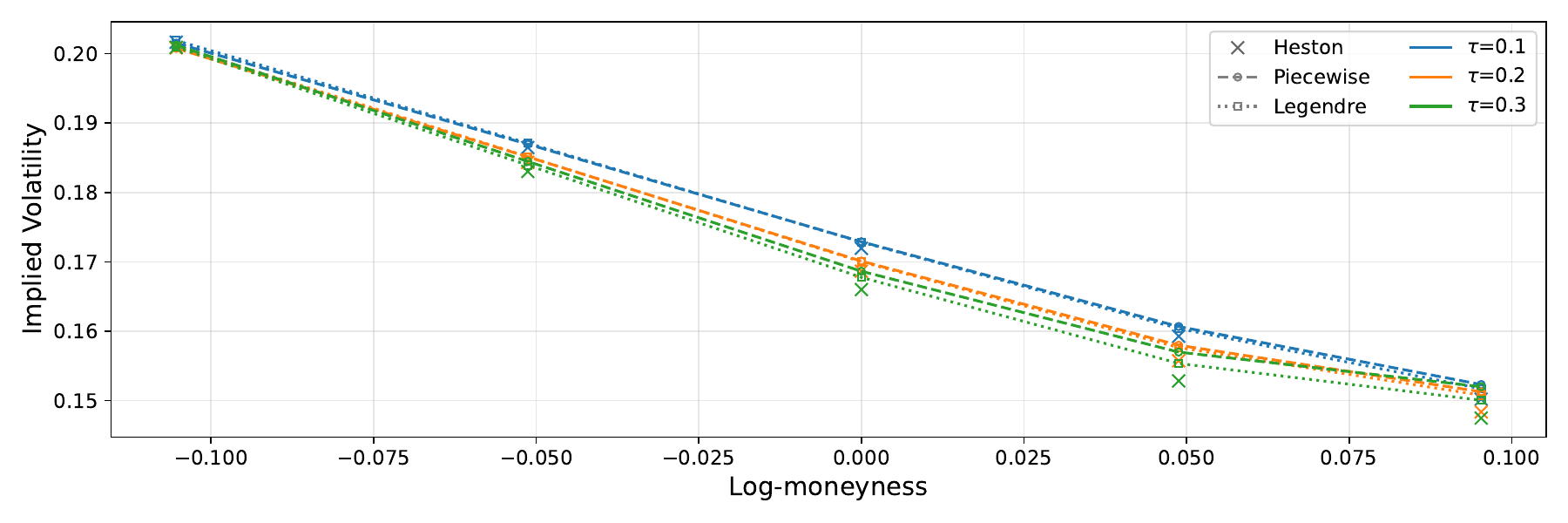}
	\end{subfigure}

	\caption{Forward-starting call option comparison, price (top) and Implied Volatility (bottom).}\label{fig:heston_fits_forward_call}
\end{figure}

\begin{figure}[!htbp]
\setlength{\parskip}{0pt}
\centering
	
	\begin{subfigure}[b]{0.90\textwidth}
		\includegraphics[width=\linewidth,trim=0 10 0 9,clip]{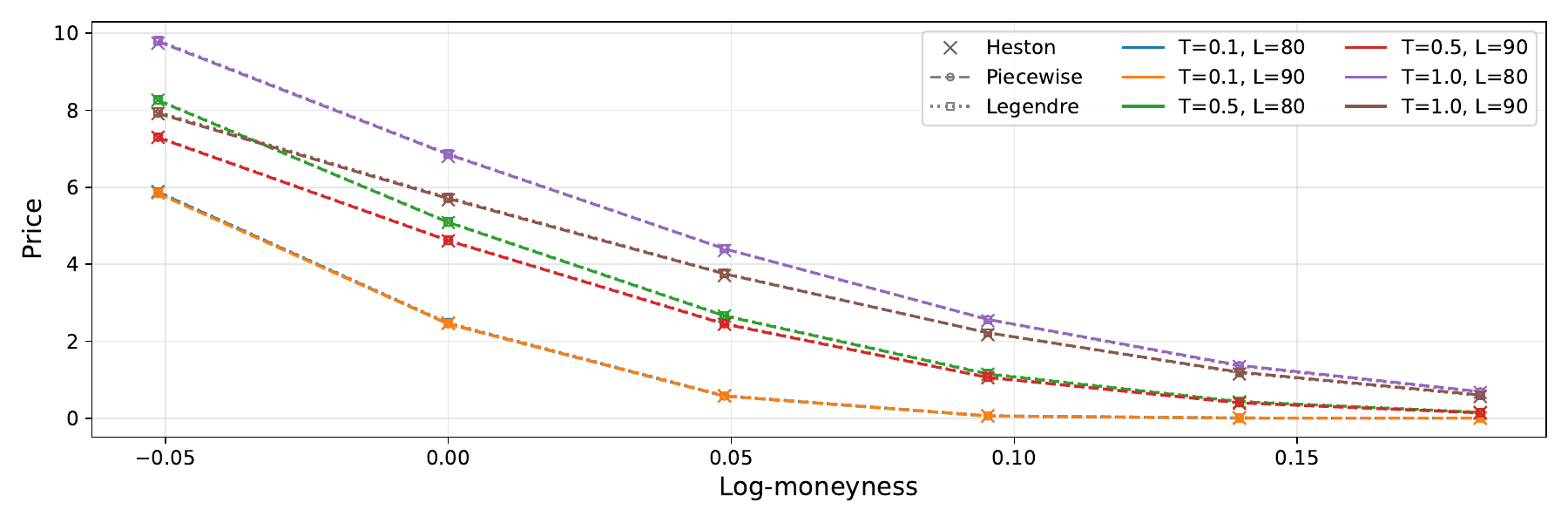}
	\end{subfigure}

    \vspace{0.1cm}
	\begin{subfigure}[b]{0.90\textwidth}
		\includegraphics[width=\linewidth,trim=0 10 0 9,clip]{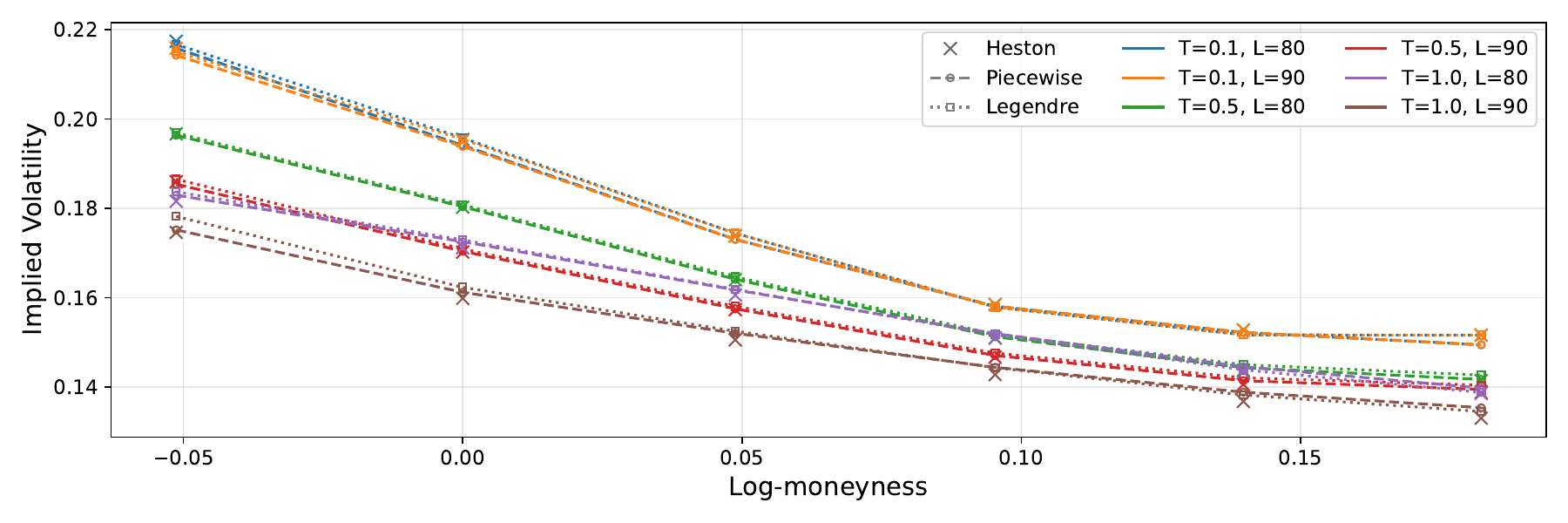}
	\end{subfigure}

	\caption{Down-and-Out call option comparison, price (top) and Implied Volatility (bottom).}\label{fig:heston_fits_barrier}
\end{figure}

\begin{figure}[!htbp]
\setlength{\parskip}{0pt}
\centering
	
	\begin{subfigure}[b]{0.90\textwidth}
		\includegraphics[width=\linewidth,trim=0 10 0 9,clip]{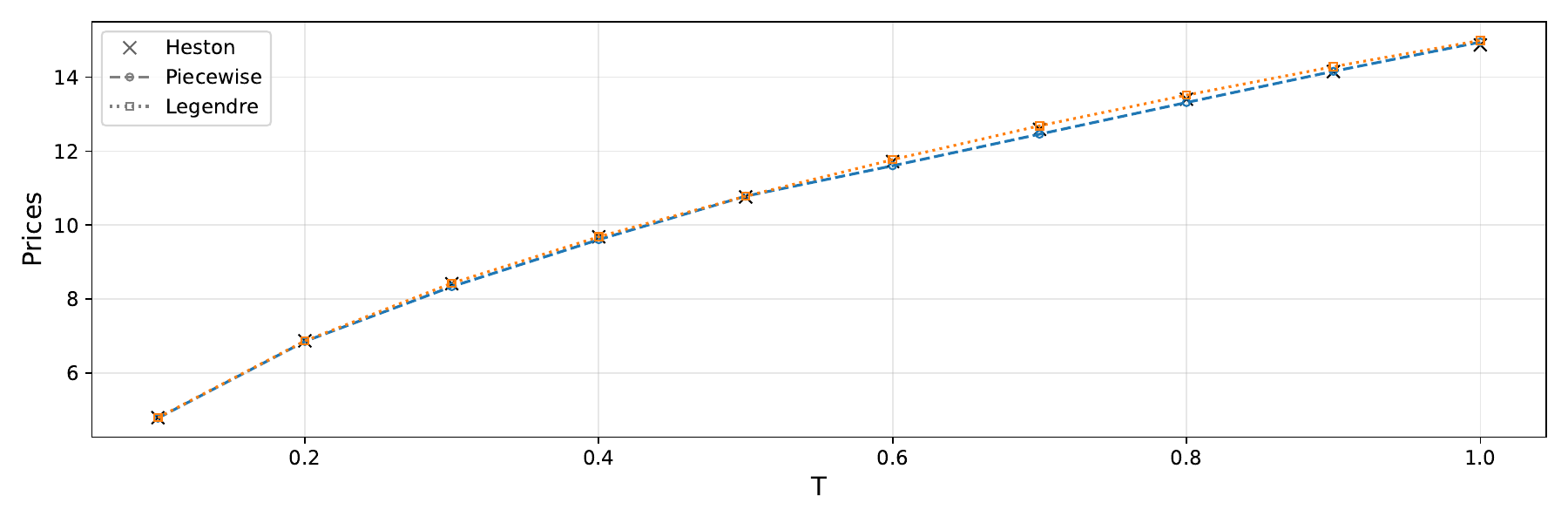}
	\end{subfigure}

    \vspace{0.1cm}
	\begin{subfigure}[b]{0.90\textwidth}
		\includegraphics[width=\linewidth,trim=0 10 0 9,clip]{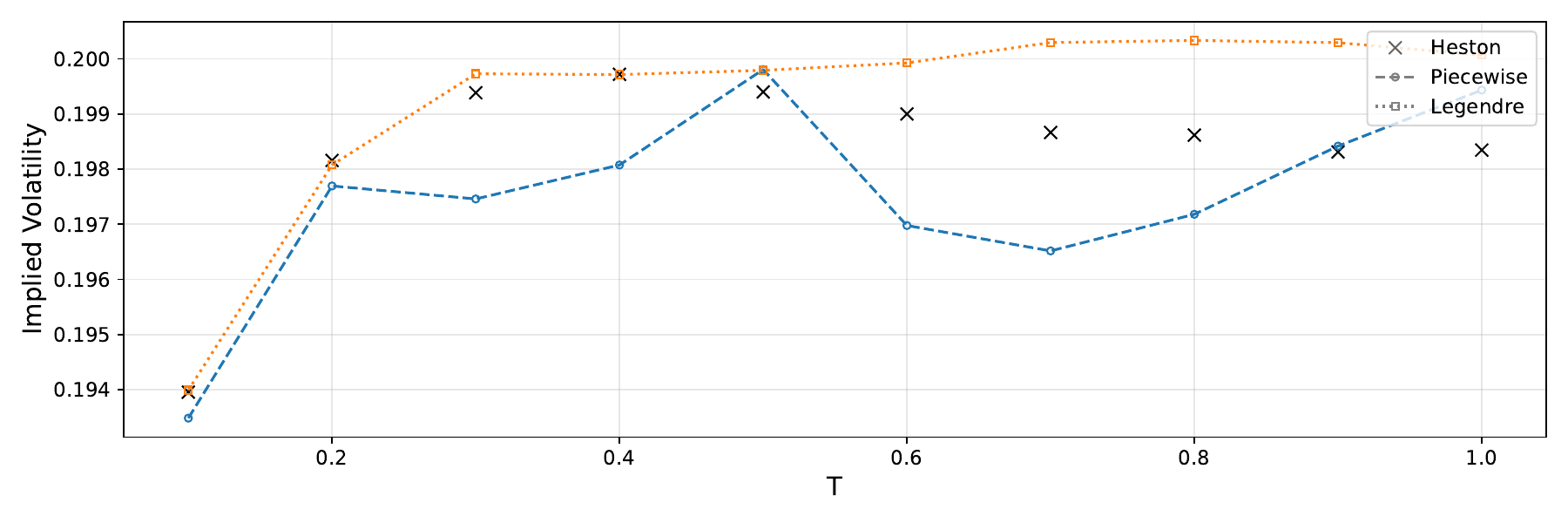}
	\end{subfigure}

	\caption{Floating-strike lookback call options comparison, price (top) and Implied Volatility (bottom).}\label{fig:heston_fits_lookback}
\end{figure}

Figures~\ref{fig:heston_fits_forward_call}--\ref{fig:heston_fits_lookback} compare prices and implied volatilities for the three exotic contracts. For forward-starting calls (maturity one year and start times $\tau\in\{0.1,0.2,0.3\}$), both bases match Heston prices well across strikes, with deviations in implied volatility becoming more noticeable for larger $\tau$.

For down-and-out calls, both bases closely match prices and implied volatilities across strikes for the maturity--barrier combinations considered.

For floating-strike lookback calls, prices remain close to Heston for both bases, but differences are more apparent in implied volatility. The piecewise-constant basis tends to underestimate Heston implied volatilities at intermediate maturities and slightly overestimates them at the longest maturity. The Legendre basis matches the shorter maturities closely but remains above Heston at longer maturities, failing to reproduce the gradual decline in the reference implied volatilities.

Overall, the chaos-based model provides a reasonable approximation of the exotic option prices considered, suggesting that the learned dynamics capture useful features beyond the European option surface used in calibration.

\subsection{The rough Heston model}
The rough Heston model of \textcite{RoughHeston1,RoughHeston2,AffineVolterraProcesses} is a stochastic volatility model in which the log-price $X=\log(S)$ has risk-neutral dynamics
\begin{flalign*}
       dX_t &= -\tfrac{1}{2}V_t\,dt + \sqrt{V_t}\Big(\rho\, dB_{t}^{1}+\sqrt{1-\rho^2}\,dB_t^{2}\Big), \\
        V_t &= V_0
    + \frac{\kappa}{\Gamma(\alpha)}
      \int_0^t (t-s)^{\alpha-1}(\bar{v}-V_s)\,ds
    + \frac{\varepsilon}{\Gamma(\alpha)}
      \int_0^t (t-s)^{\alpha-1}\sqrt{V_s}\,dB_s^1,
\end{flalign*}
where $B^{1}$ and $B^{2}$ are independent Brownian motions and $\kappa,\bar{v},\varepsilon>0$ play roles analogous to those in the classical Heston model.

The key difference from the Heston model is the convolution kernel $(t-s)^{\alpha-1}/\Gamma(\alpha)$ with $\alpha\in(1/2,1]$;
the case $\alpha=1$ reduces to the classical Heston model. The kernel introduces memory into the variance process, so the model is no longer Markovian. When $\alpha<1$, the variance process $V$ has sample paths that are almost surely locally H\"older continuous of every order $\gamma\in(0,\alpha-\tfrac{1}{2})$. It has been argued that such roughness is important to reproduce empirically observed features of implied-volatility surfaces \citep[see, e.g.,][]{Fukasawa2011,RoughVolPricing}.

The characteristic function of $X_t$ can be computed using the fractional Riccati equation in \textcite{RoughHeston1}, with efficient rational approximations available in \textcite{RationalApproximations}.

\subsubsection{Fit to option implied volatilities}

We use the parameters $(S_0,\kappa,\bar{v},\varepsilon,\rho,V_0)$ from \eqref{eq:heston_params} and set $\alpha=0.55$, corresponding to a roughness exponent $H=\alpha-\tfrac{1}{2}=0.05$. This gives highly irregular variance paths, consistent with the small roughness exponents reported in empirical studies.

We calibrate to call option prices on a grid of ten maturities, from one week to 18 months, and nine strikes spanning a moneyness range of $\pm 20\%$. Compared with the Heston example, these shorter maturities make calibration more challenging because of the pronounced short-maturity skew generated by rough Heston, which standard stochastic volatility models struggle to reproduce.

For the Wiener chaos martingale model, we set $\mathcal{T}=1.5$ and use $d=2$ Brownian drivers. Both bases use $M=10$ functions per driver and total chaos order $P=3$. The piecewise-constant basis consists of normalized indicators on the ten nonuniform intervals defined by zero and the calibration maturities. Each basis gives 1,771 coefficients, including the constant term, and 1,770 effective continuous degrees of freedom after normalization.

Figure~\ref{fig:rough_heston_fits} compares the rough Heston implied-volatility surface with the fitted surfaces for both bases, with error statistics reported in Table~\ref{tab:error_stats_rough_heston}. 

Figure~\ref{fig:rough_heston_short_maturity} focuses on the three shortest maturities (one week, two weeks, and one month), where both bases closely reproduce the reference smiles across the strike range. Figure~\ref{fig:rough_heston_fits_oos} shows the corresponding fits at uncalibrated maturities, where the lower average MAE of the Legendre basis again suggests better generalization beyond the calibration grid.

\begin{table}[!htbp]
\centering
\begin{tabular}{|l|cc|c|}
\hline
& \multicolumn{2}{c|}{MAE (bps)}
& Calibration Time (s) \\
\hline
Basis
& Calibrated & Uncalibrated
& Average $\pm$ Std \\
\hline
piecewise-constant
& $5.36 \pm 1.02$ & $17.98 \pm 1.90$ & $61.27 \pm 11.01$ \\
Legendre
& $10.26 \pm 1.72$ & $9.22 \pm 1.40$ & $95.29 \pm 29.92$ \\
\hline
\end{tabular}
\caption{Statistics over 100 calibrations per basis to Rough Heston option prices. We report the mean absolute error (MAE) of the implied-volatility surfaces in basis points (average $\pm$ standard deviation) for calibrated and uncalibrated maturities. For the Legendre basis, the shortest uncalibrated maturity is priced using five independent batches of 4 million Monte Carlo paths. Calibration time includes Heston initialization and chaos calibration.}
\label{tab:error_stats_rough_heston}
\end{table}

\begin{figure}[!htbp]
\setlength{\parskip}{0pt}
\centering
	
	\begin{subfigure}[b]{0.48\textwidth}
		\includegraphics[width=\linewidth]{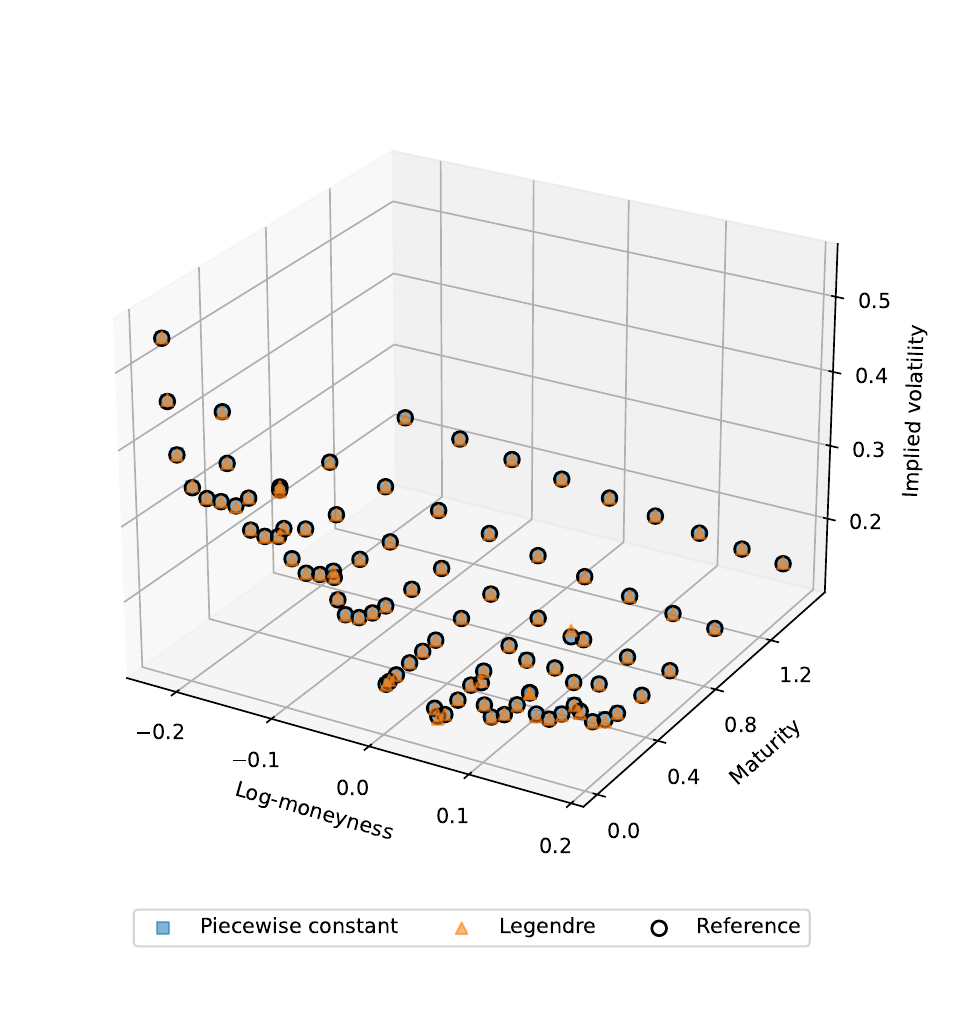}
	\end{subfigure}
	\begin{subfigure}[b]{0.48\textwidth}
		\includegraphics[width=\linewidth]{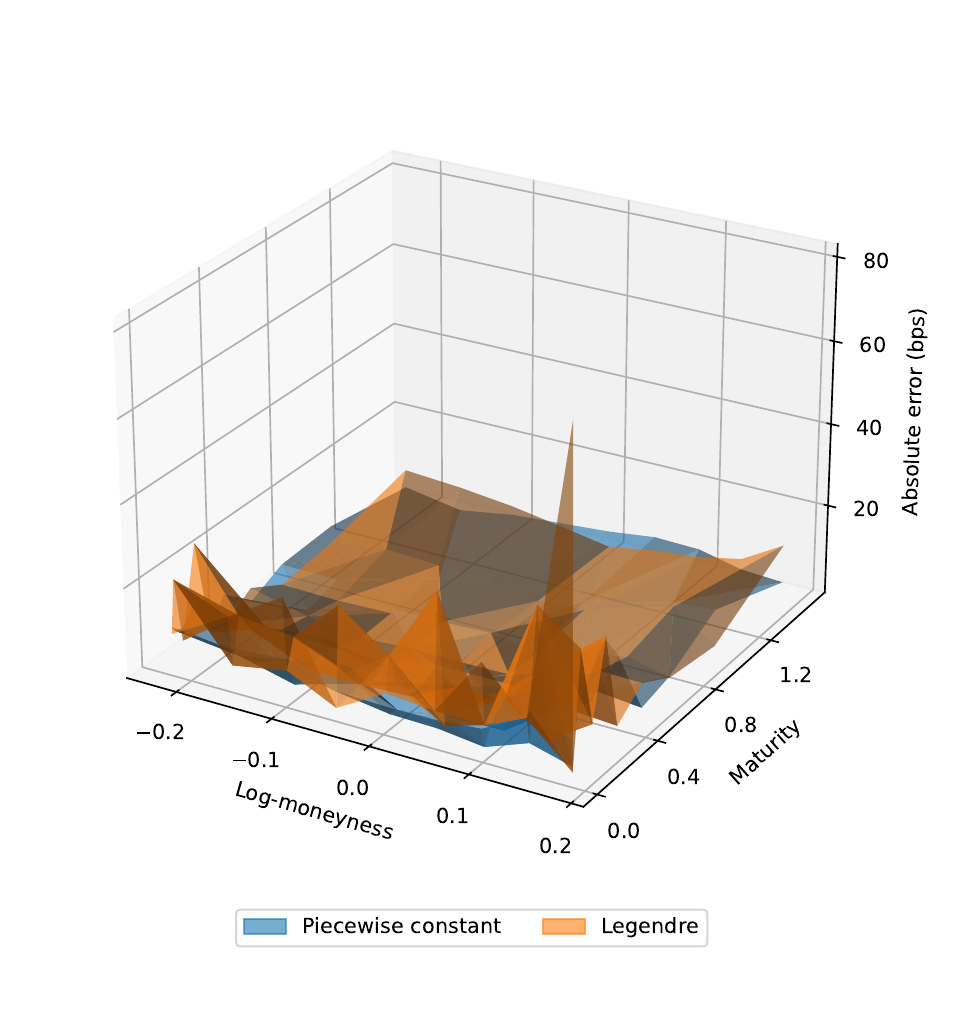}
	\end{subfigure}
	
	\vspace{0.1cm}

	\begin{subfigure}[b]{0.96\textwidth}
		\includegraphics[width=\linewidth]{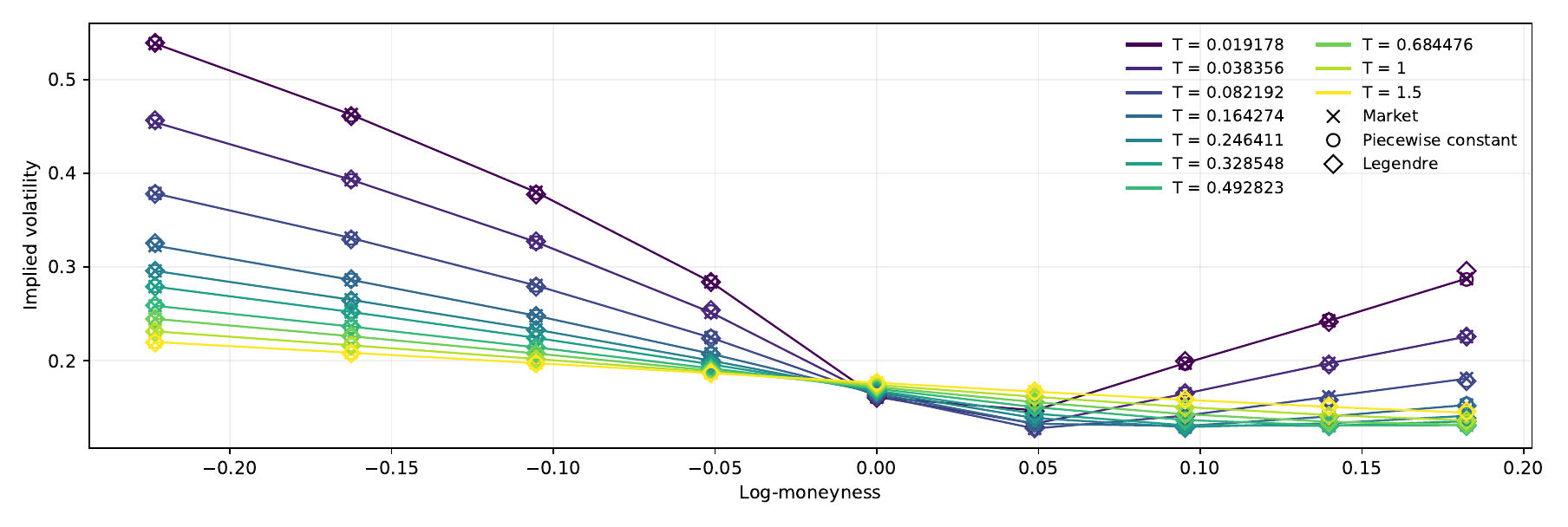}
	\end{subfigure}

	\caption{Implied-volatility surfaces at the calibrated maturities in the rough Heston model (top left) and surfaces of absolute errors (top right). The bottom panel shows the individual implied-volatility smiles. The MAE is 5.27 bps for the piecewise-constant basis and 8.87 bps for the Legendre basis.}\label{fig:rough_heston_fits}
\end{figure}

\begin{figure}[!htbp]
\setlength{\parskip}{0pt}
\centering
	
	\begin{subfigure}[b]{0.48\textwidth}
		\includegraphics[width=\linewidth]{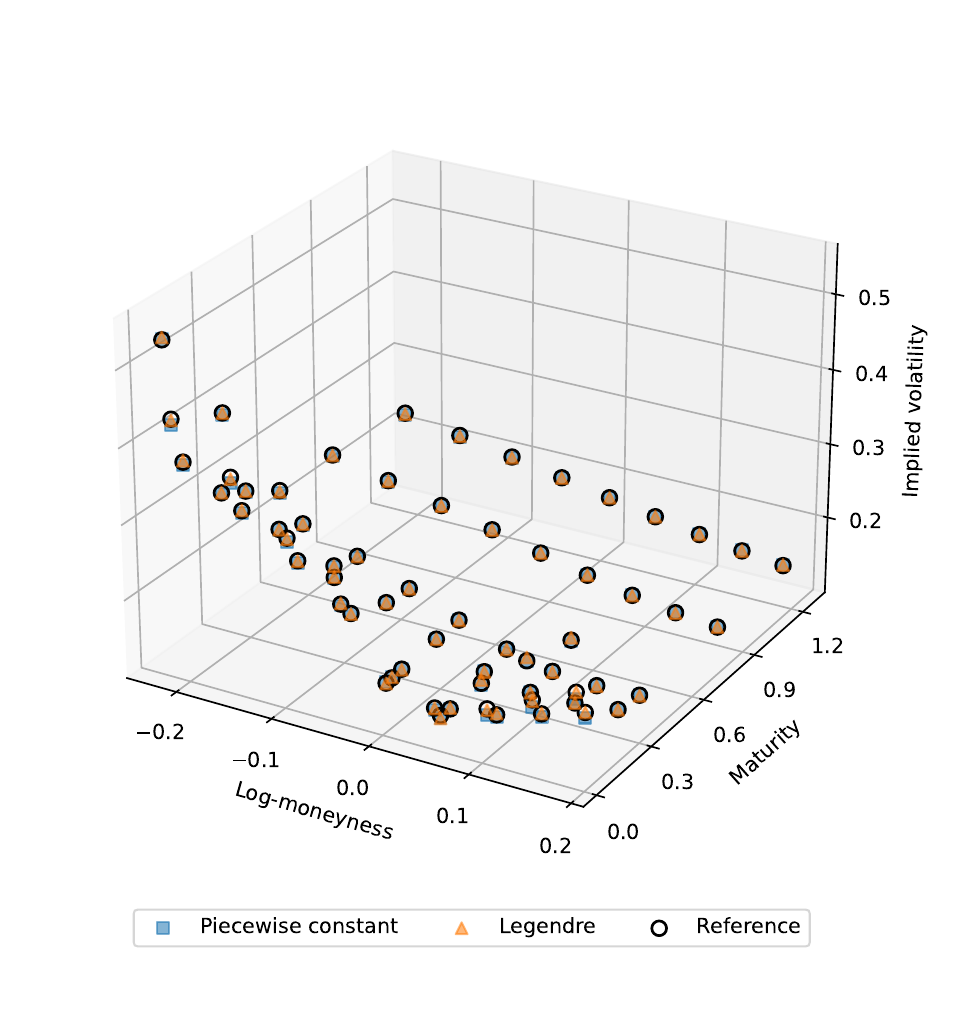}
	\end{subfigure}
	\begin{subfigure}[b]{0.48\textwidth}
		\includegraphics[width=\linewidth]{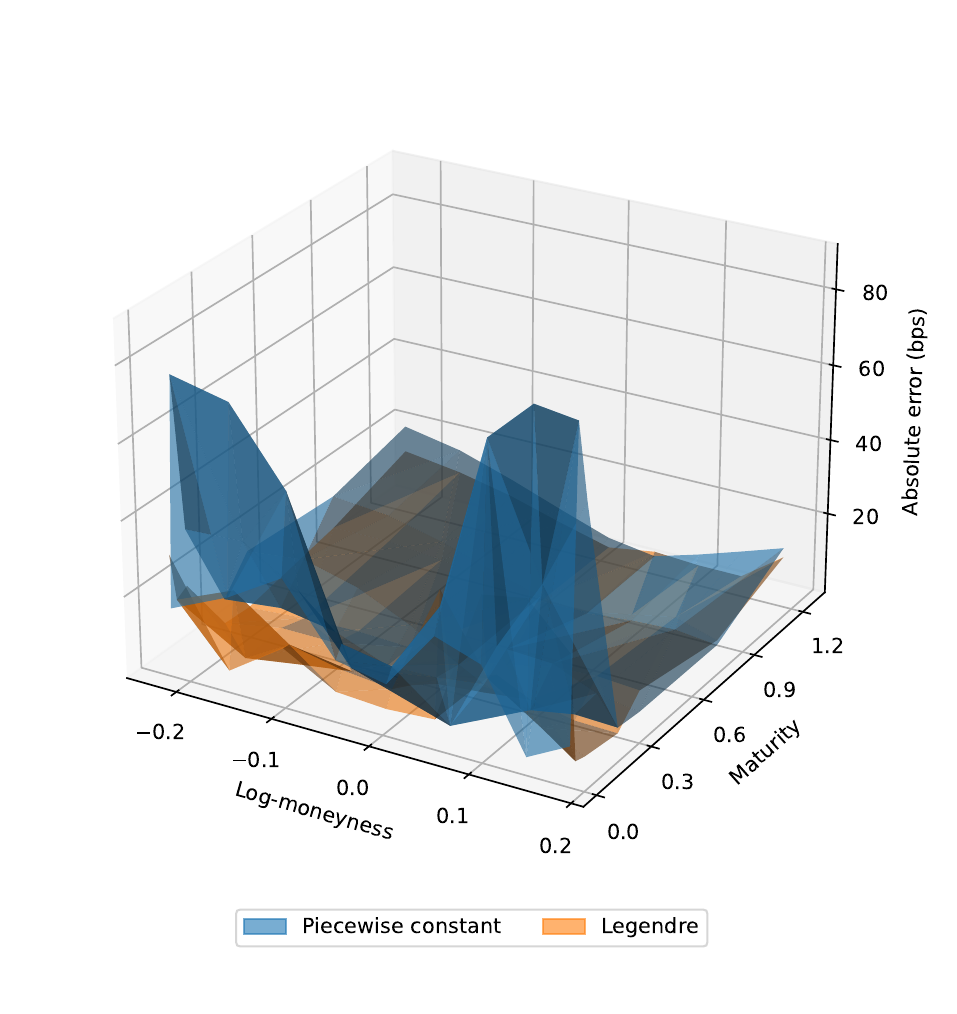}
	\end{subfigure}
	
	\vspace{0.1cm}

	\begin{subfigure}[b]{0.96\textwidth}
		\includegraphics[width=\linewidth]{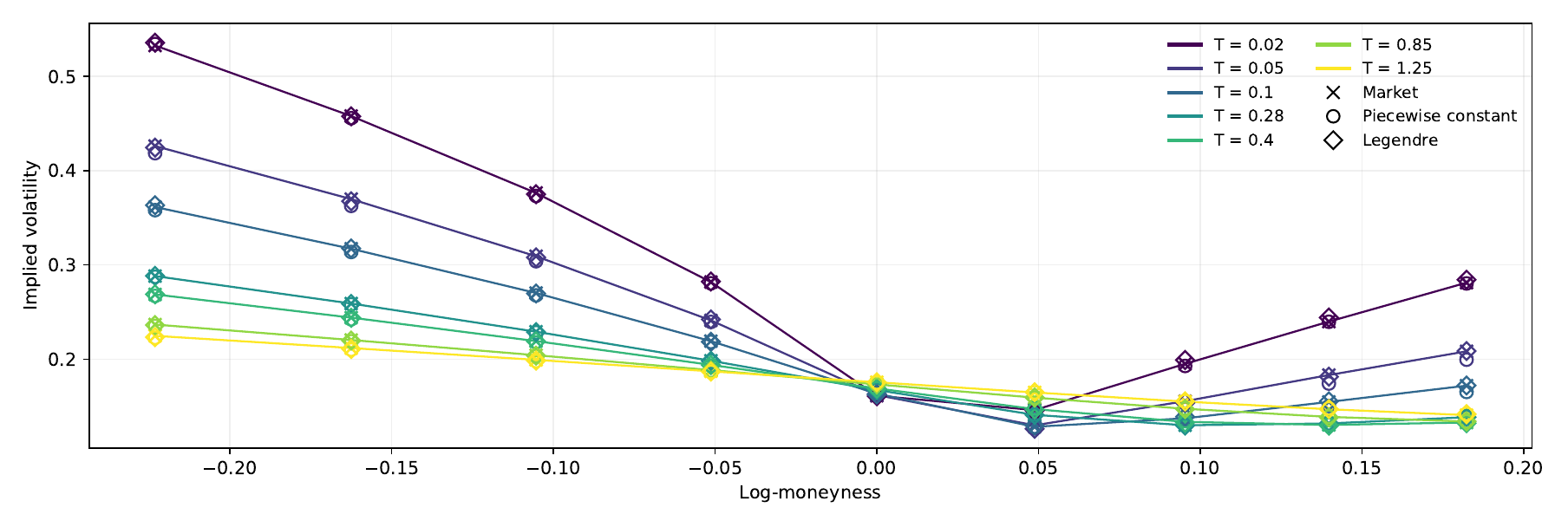}
	\end{subfigure}

	\caption{Implied-volatility surfaces at the uncalibrated maturities in the rough Heston model (top left) and surfaces of absolute errors (top right). The bottom panel shows the individual implied-volatility smiles. The MAE is 18.50 bps for the piecewise-constant basis and 8.51 bps for the Legendre basis.}\label{fig:rough_heston_fits_oos}
\end{figure}

\begin{figure}[!htbp]
\setlength{\parskip}{0pt}
\centering
	
    \textbf{T = 1 week}\par
	\begin{subfigure}[b]{0.88\textwidth}
        \centering
		\includegraphics[width=\linewidth,trim=0 10 0 10,clip]{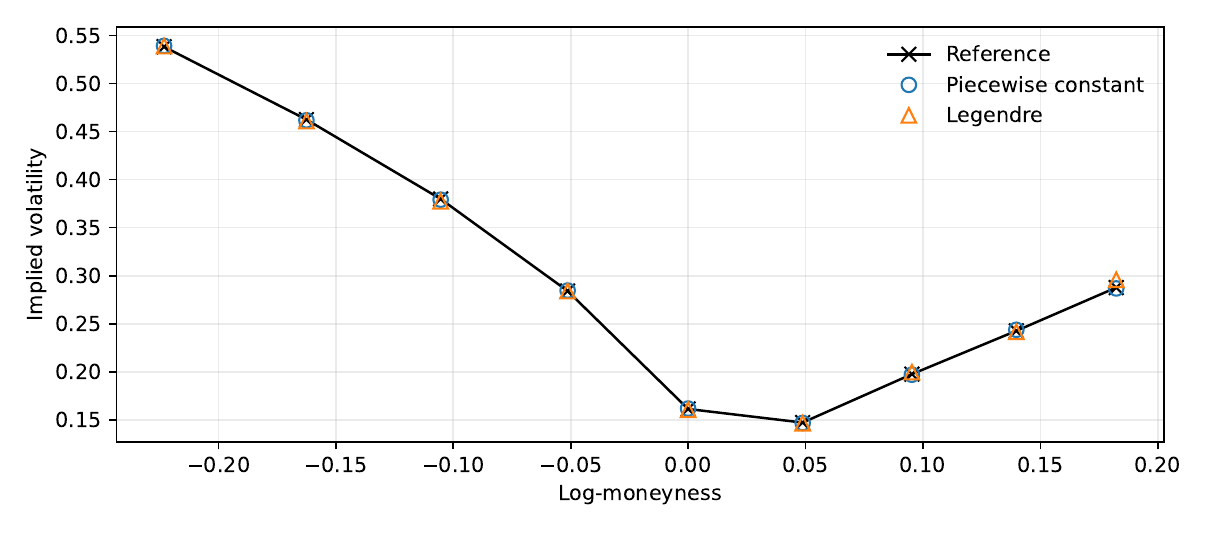}
	\end{subfigure}

    \vspace{0.1cm}

    \textbf{T = 2 weeks}\par
	\begin{subfigure}[b]{0.88\textwidth}
        \centering
		\includegraphics[width=\linewidth,trim=0 10 0 10,clip]{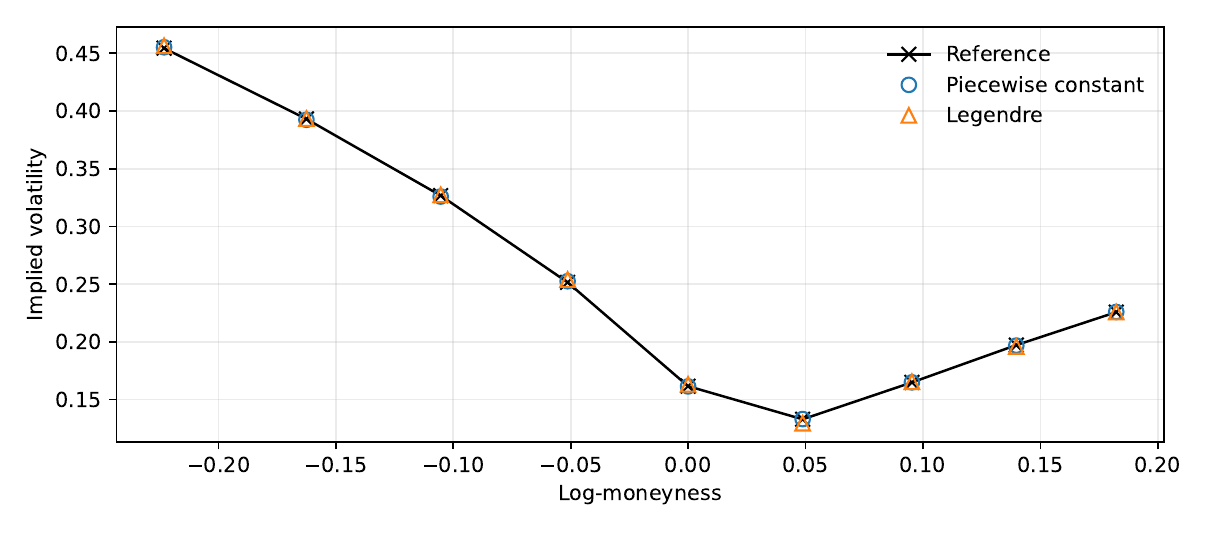}
	\end{subfigure}
	
	\vspace{0.1cm}

    \textbf{T = 1 month}\par
	\begin{subfigure}[b]{0.88\textwidth}
        \centering
		\includegraphics[width=\linewidth,trim=0 10 0 10,clip]{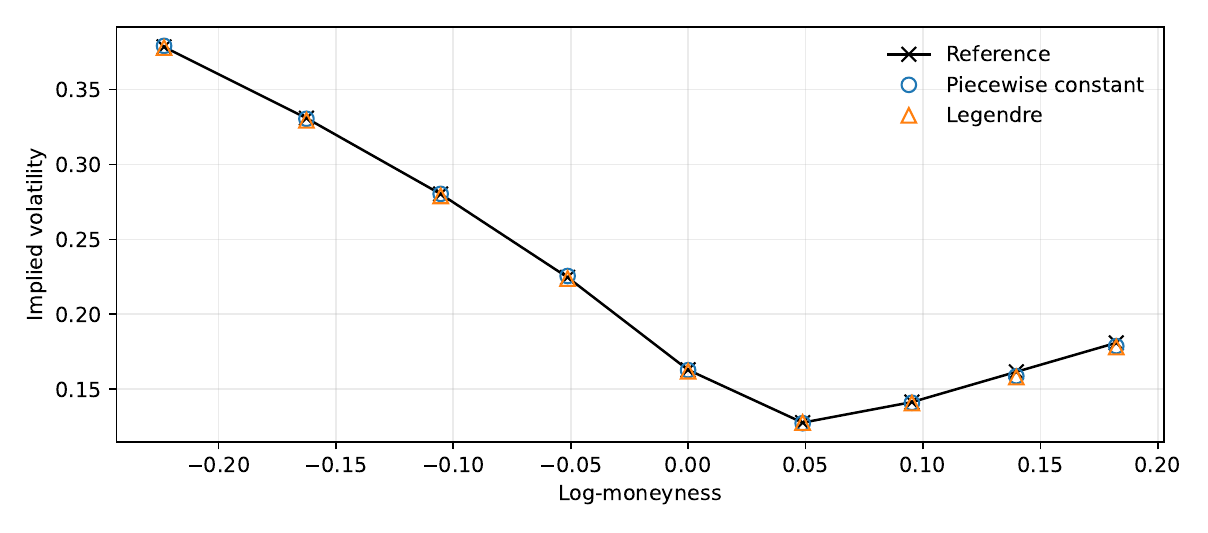}
	\end{subfigure}

	\caption{Implied-volatility fits to short maturities in the rough Heston model.}\label{fig:rough_heston_short_maturity}
\end{figure}

\subsection{Fit to SPX options data}\label{sec:data}

We conclude the numerical experiments by calibrating the model to 188 SPX call option prices from August 24, 2022, obtained from OptionMetrics. The data cover ten maturities between 9 days and 16 months, with the moneyness range varying across maturities depending on quote liquidity. We calibrate to all available maturities, incorporating maturity-dependent interest rates and a constant dividend yield.

We set $\mathcal{T}=1.30868$ and use $d=2$ Brownian drivers, $M=10$ basis functions per driver, and total chaos order $P=3$. The piecewise-constant time grid is defined by zero and the ten calibration maturities. For either basis, this gives 1,771 coefficients and 1,770 effective continuous degrees of freedom after normalization.

Figure~\ref{fig:data_fits} compares the market and fitted implied-volatility smiles, showing a close fit for both bases. Table~\ref{tab:error_stats_spx} reports the corresponding error statistics and calibration times over 100 calibrations per basis.

\begin{table}[!htbp]
\centering
\begin{tabular}{|l|c|c|}
\hline
& MAE (bps) & Calibration Time (s) \\
\hline
Basis & Calibrated & Average $\pm$ Std \\
\hline
piecewise-constant & $8.86 \pm 0.60$ & $85.65 \pm 12.60$ \\
Legendre & $10.29 \pm 0.68$ & $116.83 \pm 11.37$ \\
\hline
\end{tabular}
\caption{Statistics over 100 calibrations per basis to SPX option prices. We report the mean absolute error (MAE) of the implied-volatility surface in basis points (average $\pm$ standard deviation). Calibration time includes Heston initialization and chaos calibration (average $\pm$ standard deviation).}
\label{tab:error_stats_spx}
\end{table}

\begin{figure}[!htbp]
\setlength{\parskip}{0pt}
\centering
	\centering
	\begin{subfigure}[b]{0.32\textwidth}
        \caption*{\textbf{T=9 days}}
		\includegraphics[width=\linewidth]{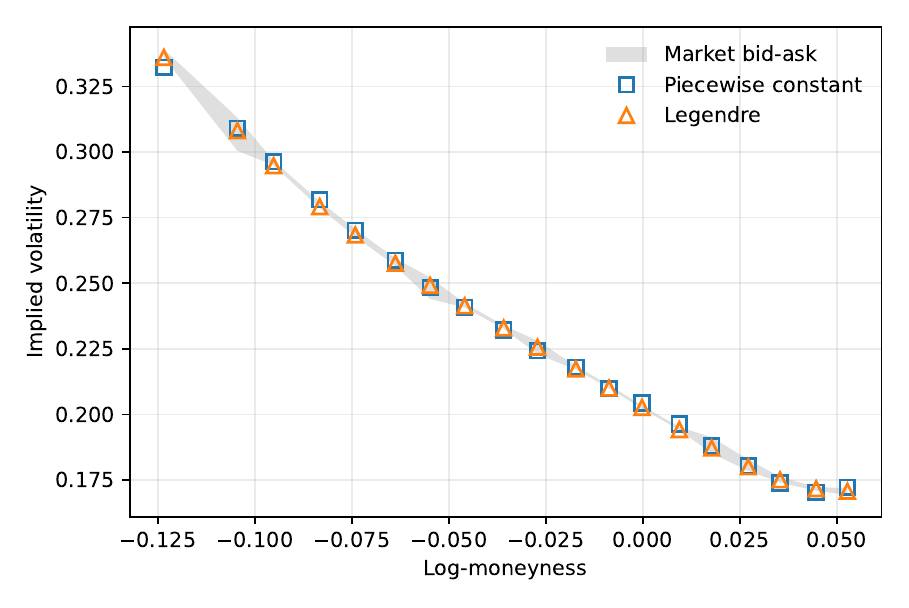}
	\end{subfigure}
	\hfill
	\begin{subfigure}[b]{0.32\textwidth}
        \caption*{\textbf{T=16 days}}
		\includegraphics[width=\linewidth]{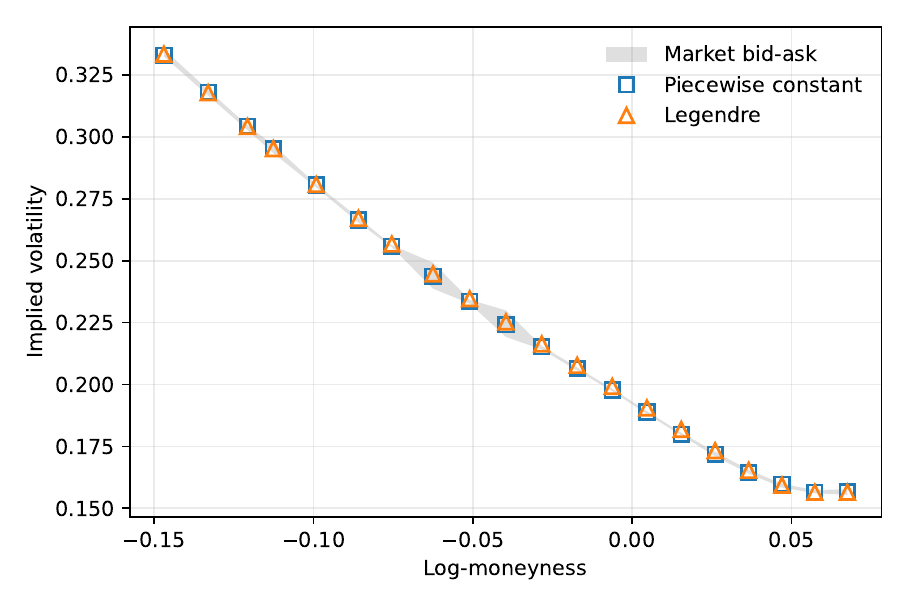}
	\end{subfigure}
	\hfill
	\begin{subfigure}[b]{0.32\textwidth}
        \caption*{\textbf{T=23 days}}
		\includegraphics[width=\linewidth]{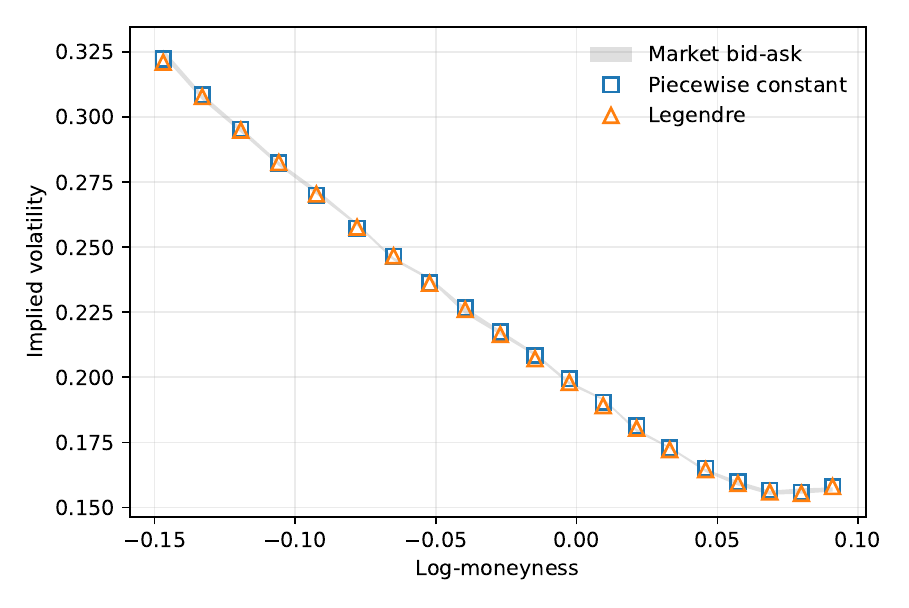}
	\end{subfigure}
    \vspace{0.05cm}
    \begin{subfigure}[b]{0.32\textwidth}
        \caption*{\textbf{T=58 days}}
		\includegraphics[width=\linewidth]{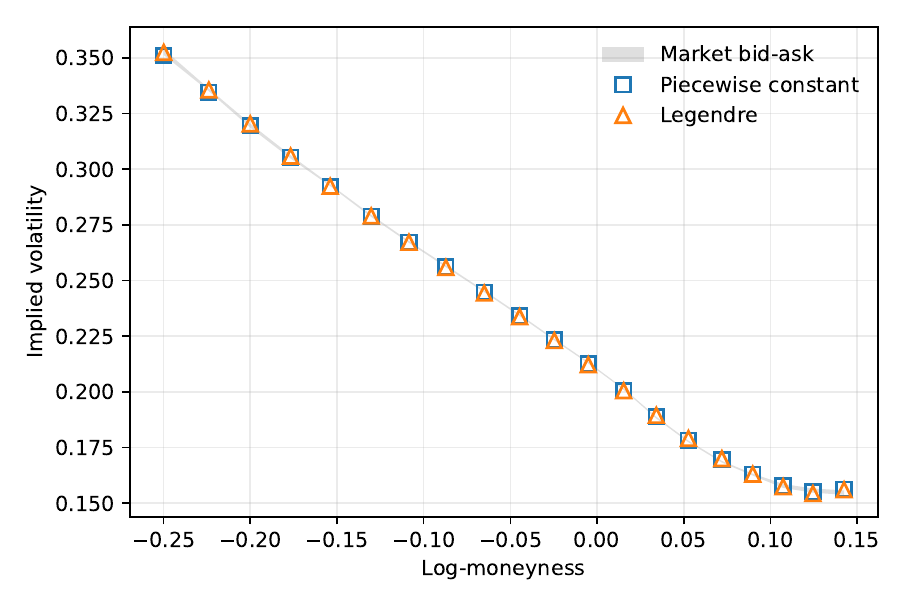}
	\end{subfigure}
	\hfill
	\begin{subfigure}[b]{0.32\textwidth}
        \caption*{\textbf{T=86 days}}
		\includegraphics[width=\linewidth]{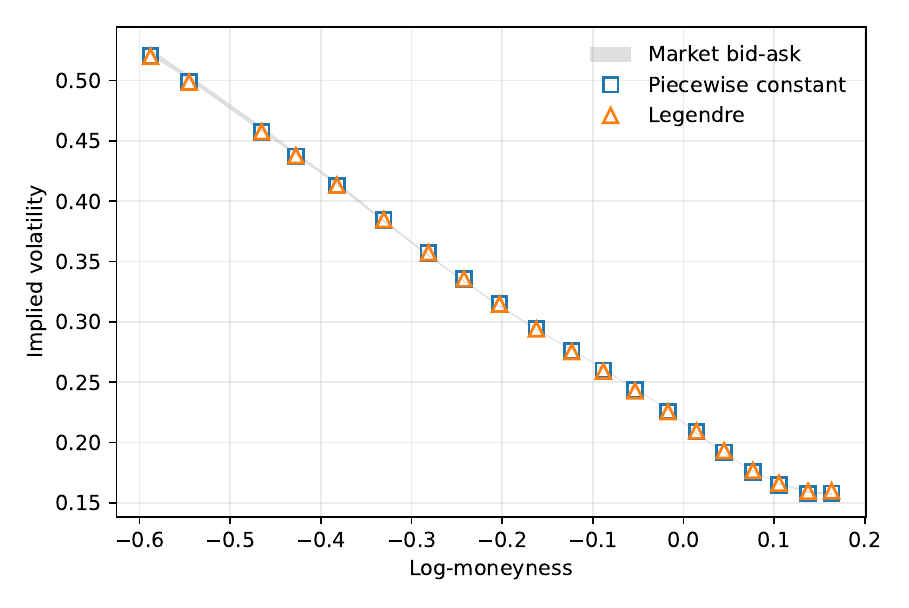}
	\end{subfigure}
	\hfill
	\begin{subfigure}[b]{0.32\textwidth}
        \caption*{\textbf{T=114 days}}
		\includegraphics[width=\linewidth]{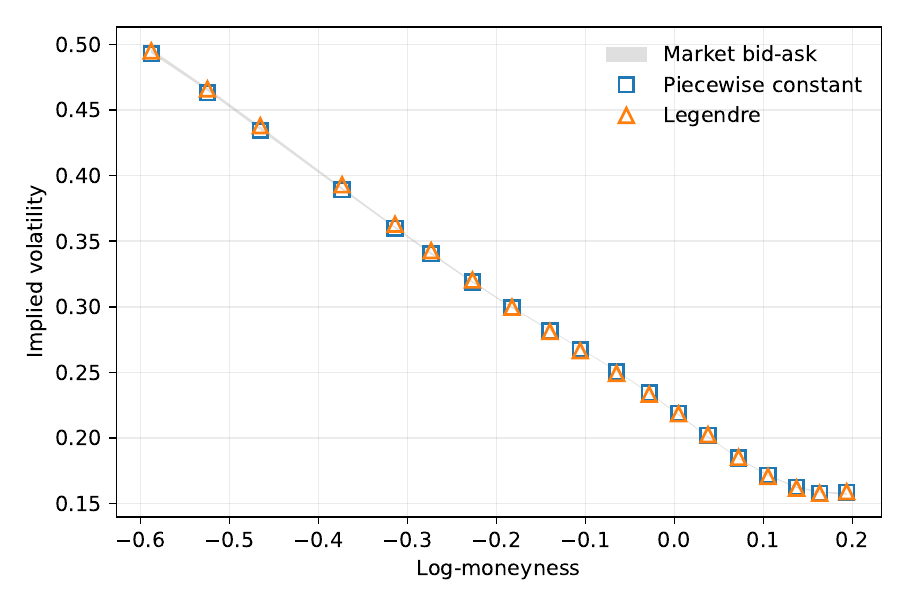}
	\end{subfigure}
    \vspace{0.05cm}
    \begin{subfigure}[b]{0.32\textwidth}
        \caption*{\textbf{T=128 days}}
		\includegraphics[width=\linewidth]{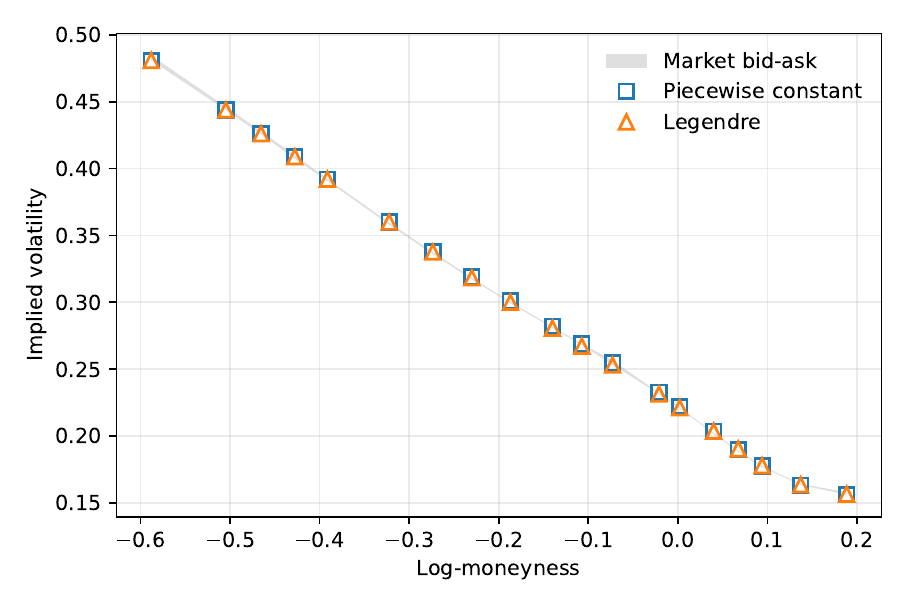}
	\end{subfigure}
	\hfill
	\begin{subfigure}[b]{0.32\textwidth}
        \caption*{\textbf{T=296 days}}
		\includegraphics[width=\linewidth]{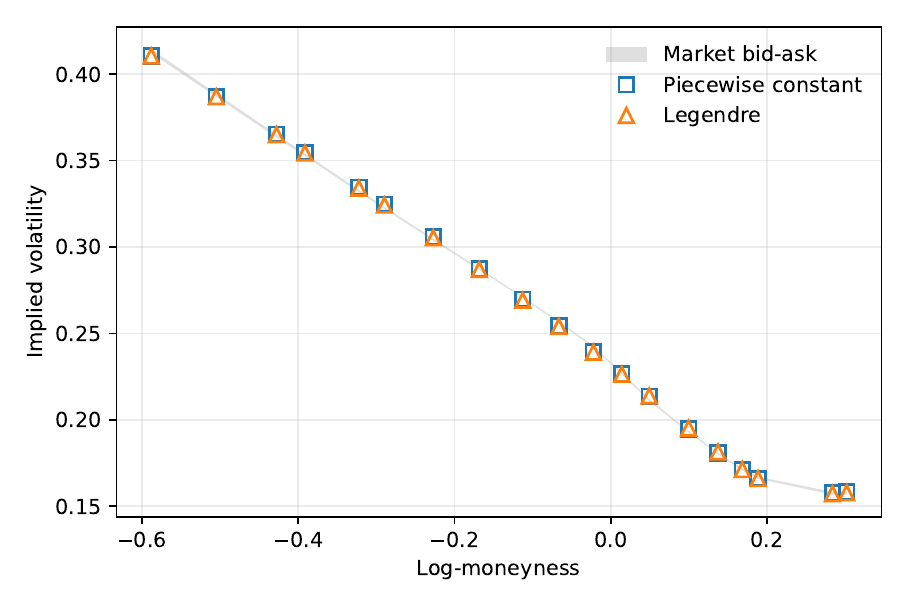}
	\end{subfigure}
	\hfill
	\begin{subfigure}[b]{0.32\textwidth}
        \caption*{\textbf{T=359 days}}
		\includegraphics[width=\linewidth]{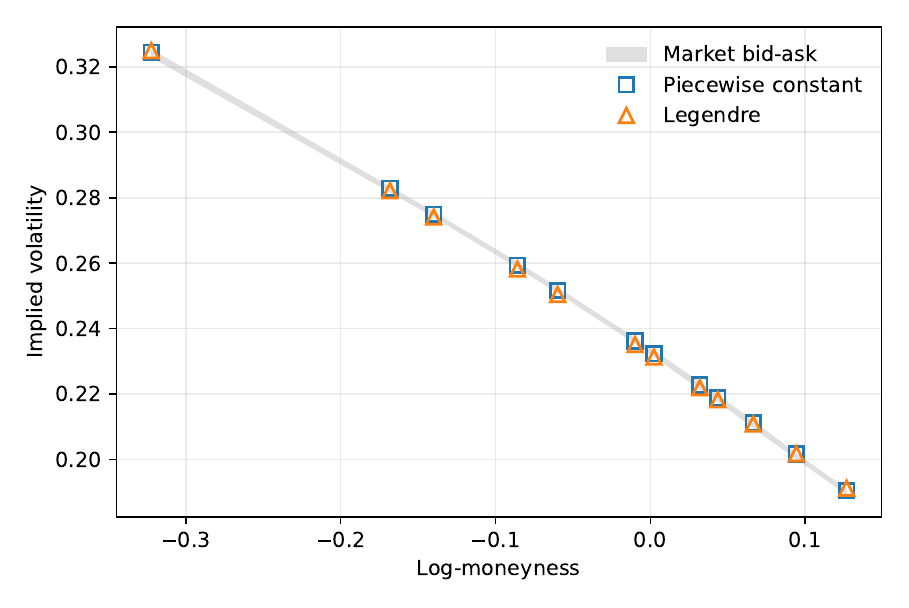}
	\end{subfigure}
	\par\vspace{0.05cm}\noindent\makebox[\textwidth][c]{%
	\begin{subfigure}[b]{0.32\textwidth}
        \caption*{\textbf{T=478 days}}
		\includegraphics[width=\linewidth]{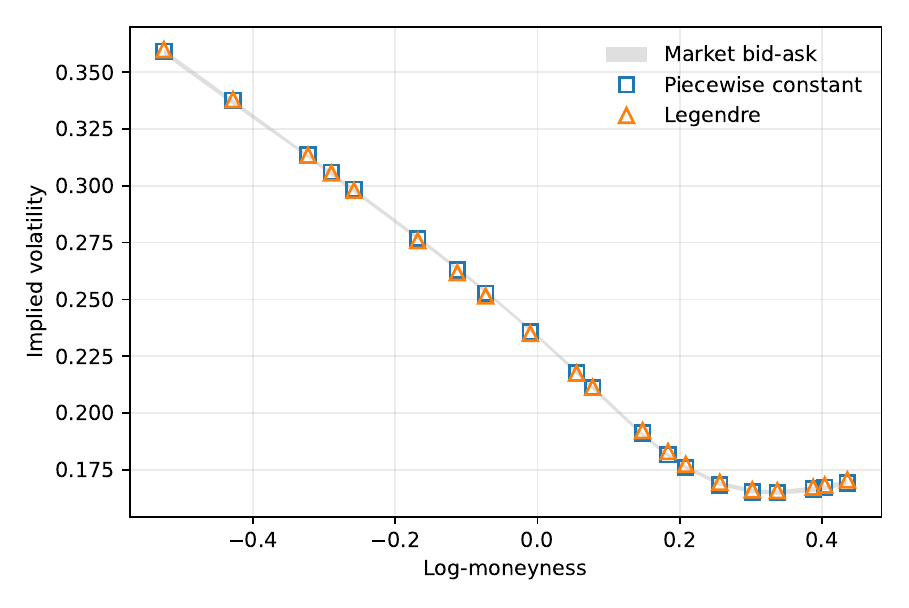}
	\end{subfigure}
	}\par
\caption{Calibrated implied-volatility smiles from SPX data. The gray band represents the observed bid--ask implied-volatility range; blue squares and orange triangles represent the fitted piecewise-constant and Legendre implied volatilities, respectively. The MAE against observed mid implied volatilities is 8.74 bps for the piecewise-constant basis and 10.69 bps for the Legendre basis.}\label{fig:data_fits}
\end{figure}

\clearpage

\section*{Funding}
Pere Diaz-Lozano gratefully acknowledges support from the SURE-AI Center, funded by the Research Council of Norway (grant no. 357482).

Thomas K. Kloster gratefully acknowledges financial support from the Center of Research in Energy: Economics and Markets and The Danish Council of Independent Research under DFF grant 10.46540/5247-00005B.
\appendix

\section{Calibration details}\label{sec:calibration-details}

In this section, we provide implementation details for the calibration procedure used in the numerical experiments of Section~\ref{sec:numerical-experiments}.

\subsection{Calibration hyperparameters}

In all cases, we initialize $\theta$ by projecting a Heston model calibrated to option prices onto the chosen basis.

As discussed in Section \ref{sec:calibration}, the parameters of the model are calibrated by minimizing the Vega-weighted mean squared error between market and model option prices using a stochastic gradient-based optimizer. In all experiments, we employ the AdamW algorithm, which uses adaptive learning rates to stabilize training in the presence of a large number of parameters.

The optimization is run for up to $10{,}000$ iterations for Heston and up to $600$ iterations for rough Heston and SPX, always with an initial learning rate of $10^{-4}$. A plateau scheduler reduces the learning rate by a factor of $0.95$ with patience $15$ and a minimum learning rate of $10^{-6}$. Weight decay is zero in all reported calibrations.

As described in Section~\ref{sec:MC}, control variates are employed in the Monte Carlo pricing approximation. For both the piecewise-constant and Legendre bases, we use control variates of polynomial degree $2$. The corresponding coefficients $\beta$ are estimated using the same Monte Carlo samples as the option prices. To further reduce overfitting of the MC estimator during calibration, the underlying Brownian paths are periodically resimulated every $50$ iterations.

A patience-based stopping criterion is used. For Heston, training stops after $30$ iterations without a loss improvement greater than $10^{-7}$, provided the best loss is below $10^{-6}$. For rough Heston and SPX, the corresponding values are $60$, $2\times10^{-8}$, and $2\times10^{-6}$, respectively.

All reported results use the parameters selected by the lowest monitored training loss, up to the stated tolerance. The reported Monte Carlo pricing results use simulated samples not used in calibration.

We next specify, for each maturity, the corresponding details used to approximate model option prices. Recall that we propose two pricing approaches: a Monte Carlo estimator (MC) and Gauss--Hermite quadrature (Q). For the Monte Carlo method, we report the number of simulated paths used for the reported independent pricing evaluation, while for Gauss--Hermite quadrature we report the number of nodes used. In the tables, the suffixes ``k'' and ``M'' denote thousands and millions, respectively.


\begin{table}[!htbp]
\centering
\scriptsize
\setlength{\tabcolsep}{3pt}
\renewcommand{\arraystretch}{1.15}

\begin{tabularx}{\linewidth}{l*{7}{>{\centering\arraybackslash}X}}
\toprule
Basis
& \makecell{$T_1$\\{\tiny(0.0821)}}
& \makecell{$T_2$\\{\tiny(0.1725)}}
& \makecell{$T_3$\\{\tiny(0.2491)}}
& \makecell{$T_4$\\{\tiny(0.4983)}}
& \makecell{$T_5$\\{\tiny(0.9884)}}
& \makecell{$T_6$\\{\tiny(1.4867)}}
& \makecell{$T_7$\\{\tiny(1.974)}} \\
\midrule
Piecewise-constant & \Q{40} & \Q{25} & \MC{100k} & \MC{100k} & \MC{100k} & \MC{100k} & \MC{100k} \\
Legendre           & \MC{1M} & \MC{500k} & \MC{100k} & \MC{100k} & \MC{100k} & \MC{100k} & \MC{100k} \\
\bottomrule
\end{tabularx}

\caption{Pricing method used for the calibrated maturities in the Heston example.}
\label{tab:pricing_methods_heston_calibrated}
\end{table}

\begin{table}[!htbp]
\centering
\scriptsize
\setlength{\tabcolsep}{3pt}
\renewcommand{\arraystretch}{1.15}

\begin{tabularx}{\linewidth}{l*{6}{>{\centering\arraybackslash}X}}
\toprule
Basis
& \makecell{$T_1$\\{\tiny(0.13)}}
& \makecell{$T_2$\\{\tiny(0.21)}}
& \makecell{$T_3$\\{\tiny(0.35)}}
& \makecell{$T_4$\\{\tiny(0.75)}}
& \makecell{$T_5$\\{\tiny(1.25)}}
& \makecell{$T_6$\\{\tiny(1.75)}} \\
\midrule
Piecewise-constant & \Q{30} & \Q{15} & \MC{100k} & \MC{100k} & \MC{100k} & \MC{100k} \\
Legendre           & \MC{1M} & \MC{1M} & \MC{100k} & \MC{100k} & \MC{100k} & \MC{100k} \\
\bottomrule
\end{tabularx}

\caption{Pricing method used for the uncalibrated maturities in the Heston example.}
\label{tab:pricing_methods_heston_non_calibrated}
\end{table}


\begin{table}[!htbp]
\centering
\scriptsize
\setlength{\tabcolsep}{3pt}
\renewcommand{\arraystretch}{1.15}
\renewcommand{\MC}[1]{\makecell{MC\\#1}}

\begin{tabularx}{\linewidth}{l*{10}{>{\centering\arraybackslash}X}}
\toprule
Basis
& \makecell{$T_1$\\{\tiny(0.0191)}}
& \makecell{$T_2$\\{\tiny(0.0383)}}
& \makecell{$T_3$\\{\tiny(0.0821)}}
& \makecell{$T_4$\\{\tiny(0.1642)}}
& \makecell{$T_5$\\{\tiny(0.2464)}}
& \makecell{$T_6$\\{\tiny(0.3285)}}
& \makecell{$T_7$\\{\tiny(0.4928)}}
& \makecell{$T_8$\\{\tiny(0.6844)}}
& \makecell{$T_9$\\{\tiny(1.0000)}}
& \makecell{$T_{10}$\\{\tiny(1.5000)}}\\
\midrule
Piecewise-constant & \Q{40} & \Q{25} & \MC{200k} & \MC{200k} & \MC{100k} & \MC{100k} & \MC{100k} & \MC{100k} & \MC{100k} & \MC{100k}\\
Legendre           & \MC{4M} & \MC{1M} & \MC{200k} & \MC{200k} & \MC{100k} & \MC{100k} & \MC{100k} & \MC{100k} & \MC{100k} & \MC{100k}\\
\bottomrule
\end{tabularx}

\caption{Pricing method used for the calibrated maturities in the rough Heston example.}
\label{tab:pricing_methods_rough_heston_calibrated}
\end{table}

\begin{table}[!htbp]
\centering
\scriptsize
\setlength{\tabcolsep}{3pt}
\renewcommand{\arraystretch}{1.15}

\begin{tabularx}{\linewidth}{l*{7}{>{\centering\arraybackslash}X}}
\toprule
Basis
& \makecell{$T_1$\\{\tiny(0.02)}}
& \makecell{$T_2$\\{\tiny(0.05)}}
& \makecell{$T_3$\\{\tiny(0.10)}}
& \makecell{$T_4$\\{\tiny(0.28)}}
& \makecell{$T_5$\\{\tiny(0.40)}}
& \makecell{$T_6$\\{\tiny(0.85)}}
& \makecell{$T_7$\\{\tiny(1.25)}} \\
\midrule
Piecewise-constant & \Q{30} & \Q{15} & \MC{1M} & \MC{500k} & \MC{500k} & \MC{500k} & \MC{500k} \\
Legendre           & \MC{5$\times$4M} & \MC{2M} & \MC{2M} & \MC{500k} & \MC{500k} & \MC{500k} & \MC{500k} \\
\bottomrule
\end{tabularx}

\caption{Pricing method used for the uncalibrated maturities in the rough Heston example.}
\label{tab:pricing_methods_rough_heston_non_calibrated}
\end{table}


\begin{table}[!htbp]
\centering
\scriptsize
\setlength{\tabcolsep}{2pt} 
\renewcommand{\arraystretch}{1.15}

\begin{tabularx}{\linewidth}{l*{10}{>{\centering\arraybackslash}X}}
\toprule
Basis
& \makecell{$T_1$\\{\tiny(0.0246)}}
& \makecell{$T_2$\\{\tiny(0.0438)}}
& \makecell{$T_3$\\{\tiny(0.0629)}}
& \makecell{$T_4$\\{\tiny(0.1587)}}
& \makecell{$T_5$\\{\tiny(0.2354)}}
& \makecell{$T_6$\\{\tiny(0.3121)}}
& \makecell{$T_7$\\{\tiny(0.3504)}}
& \makecell{$T_8$\\{\tiny(0.8104)}}
& \makecell{$T_9$\\{\tiny(0.9828)}}
& \makecell{$T_{10}$\\{\tiny(1.3086)}}
\\
\midrule
Piecewise-constant & \Q{40} & \Q{25} & \MC{200k} & \MC{200k} & \MC{100k} & \MC{100k} & \MC{100k} & \MC{100k} & \MC{100k} & \MC{100k}\\
Legendre           & \MC{200k} & \MC{200k} & \MC{200k} & \MC{200k} & \MC{100k} & \MC{100k} & \MC{100k} & \MC{100k} & \MC{100k} & \MC{100k}\\
\bottomrule
\end{tabularx}
\caption{Pricing method used for the calibrated maturities in the SPX example.}
\label{tab:pricing_methods_real data}
\end{table}

\FloatBarrier

\section{Exotic option prices and implied volatilities}\label{sec:exotic-iv}
In the numerical experiments with the Heston model, we consider how well the calibrated model matches implied volatilities of some path-dependent exotic options to which the model is \emph{not} calibrated. We consider the following three types of exotic options, which all admit closed form prices in the Black-Scholes model. Up to strike, maturity, interest rate, and dividend yield, these prices all depend only on the model parameter $\sigma$, and we can thus define the corresponding exotic option implied volatilities by numerical inversion in $\sigma$, once the prices are computed via simulation.
\begin{enumerate}
   \item \textbf{Forward-starting call options:} These are simply call options that start at a pre-specified future date, $\tau$, and their prices are therefore sensitive to the forward evolution of the model-implied volatility surface. The strike is set at time $\tau$ to $KS_\tau/S_0$, giving payoff $(S_T-KS_\tau/S_0)_+$ at maturity $T>\tau$. The Black-Scholes price of a forward-starting option is
\[
C_{BS}^{\mathrm{fwd}}(S_0,T,K,\tau;\sigma) = e^{-q\tau}C_{BS}(S_0,T-\tau,K;\sigma),
\]
where $C_{BS}$ denotes the Black-Scholes call option price.
    \item \textbf{Down-and-out call options:} These have payoff $(S_T-K)^+\mathbf{1}_{\lbrace \inf_{t\in [0,T]}S_{t}>L \rbrace},$ for some lower barrier $0<L<\min\{S_0,K\}$. The Black-Scholes price of a down-and-out call option is 
    \[
    C_{BS}^{\mathrm{dao}}(S_0,T,K,L;\sigma) = C_{BS}(S_0,T,K;\sigma)-\left(\frac{L}{S_0}\right)^{\frac{2(r-q-\sigma^2/2)}{\sigma^2}}C_{BS}(L^2/S_0,T,K;\sigma).
    \]
    \item \textbf{Floating-strike lookback call options:} These have payoff $(S_T-\inf_{t\in [0,T]}S_t)^{+}$. For $r\neq q$, the Black-Scholes price of a lookback call option is
    \[
    C_{BS}^{\mathrm{lb}}(S_0,T;\sigma) = C_{BS}(S_0,T,S_0;\sigma)-\frac{S_0e^{-qT}\sigma^2}{2(r-q)}\left( \Phi(-d_1)- e^{-(r-q)T}\Phi(-d_3) \right),
    \]
    where $d_1=\frac{(r-q+\frac{1}{2}\sigma^2)T}{\sigma \sqrt{T}}$ and $d_3 = d_1 - \frac{2(r-q)\sqrt{T}}{\sigma}$. For $r=q$, taking the continuous limit gives
    \[
    C_{BS}^{\mathrm{lb}}(S_0,T;\sigma)
    =
    S_0e^{-rT}\left(
    2\Phi(d_1)-1+\sigma\sqrt{T}\,\varphi(d_1)
    -\frac{\sigma^2T}{2}\Phi(-d_1)
    \right),
    \]
    where $\Phi$ and $\varphi$ denote the standard normal distribution function and density, respectively.
\end{enumerate}

\printbibliography
\end{document}